\documentclass[superscriptaddress,nofootinbib,tightenlines,preprint]{revtex4-2}
\pdfoutput=1
\usepackage[utf8]{inputenc}

\usepackage[pdftex,breaklinks,colorlinks,linktocpage,
citecolor=blue,
urlcolor=blue]{hyperref}

\usepackage{bm}
\usepackage{amssymb}
\usepackage{amsfonts}
\usepackage{mathrsfs}
\usepackage{amsmath}
\usepackage{amsthm}
\usepackage{xcolor}
\usepackage{bbold}
\usepackage{braket}
\usepackage{physics}
\usepackage{graphicx}
\usepackage{cleveref}
\usepackage{lipsum} 
\usepackage{natbib}
\usepackage{enumitem}
\usepackage{float}
\usepackage{subfigure}
\usepackage{subcaption}

\newcommand{\be}{\begin{equation}}
\newcommand{\ee}{\end{equation}}
\newcommand{\bea}{\begin{eqnarray}}
\newcommand{\eea}{\end{eqnarray}}
\newcommand{\bsubeqs}{\begin{subequations}}
\newcommand{\esubeqs}{\end{subequations}}
\newcommand{\OO}{{\rm O}}

\newcommand{\Amin}{A_\mathrm{min}}
\def\luv{\ell_{\rm UV}}
\def\bml{\begin{subequations}}
\def\blea{\bml\begin{eqnarray}}
\def\eml{\end{subequations}}
\def\elea{\end{eqnarray}\eml}
\newcommand\restr[2]{{
  \left.\kern-\nulldelimiterspace 
  #1
  \vphantom{\big|}
  \right|_{#2} 
  }}

\newtheorem{theorem}{Theorem}[section]
\newtheorem{corollary}[theorem]{Corollary}

\newtheorem{prop}[theorem]{Proposition}

\theoremstyle{definition}
\newtheorem{definition}{Definition}[section]
\newtheorem{remark}{Remark}

\Crefname{lemma}{Lemma}{Lemmas}
\Crefname{corollary}{Corollary}{}
\Crefname{proposition}{Proposition}{Propositions}
\Crefname{conjecture}{Conjecture}{Conjectures}
\Crefname{equation}{Equation}{Equations}

\DeclareMathOperator{\arcsinh}{arcsinh}

\newcommand{\nord}[1]{{:}#1{:}}

\date{\today}

\begin{document}
\count\footins = 1000

\title{Penrose inequality for integral energy conditions}

\author{Eduardo Hafemann}
\thanks{\href{mailto:eduardo.hafemann@uni-hamburg.de}{eduardo.hafemann@uni-hamburg.de}}
\affiliation{Fachbereich Mathematik, Universit{\"a}t Hamburg, Bundesstra{\ss}e 55, 20146 Hamburg, Germany}

\author{Eleni-Alexandra Kontou}
\email{eleni.kontou@kcl.ac.uk}
\affiliation{Department of Mathematics, King’s College London, Strand, London WC2R 2LS, UK}

\begin{abstract}
The classical Penrose inequality, a relation between the ADM mass and the area of any cross section of the black hole event horizon, was introduced as a test of the weak cosmic censorship conjecture: if it fails, the trapped surface is not necessarily behind the event horizon and a naked singularity could form. Since that original derivation, a variety of proofs have developed, mainly focused on the initial data formulation on maximal spacelike slices of spacetime. Most of these proofs are applicable only for classical fields, as the energy conditions required are violated in the context of quantum field theory. In this work we provide two generalizations of the Penrose inequality for spherically symmetric spacetimes: a proof of a classical Penrose inequality using initial data and an average energy condition, and a proof of a modified Penrose inequality for evaporating black holes with a connection to the weak cosmic censorship conjecture. The latter case could also be applicable to quantum fields as it uses a condition inspired by quantum energy inequalities. Finally, we provide physically motivated examples for both. 

\end{abstract}

\maketitle
\newpage

\tableofcontents

\newpage

\section{Introduction}
\label{sec:introduction}

One of the most important open problems in general relativity is whether or not the weak cosmic censorship conjecture holds. The weak cosmic censorship, conjectured by Penrose~\cite{Penrose1969} asserts that naked singularities, those not hidden behind an event horizon, cannot occur. To test this hypothesis, Penrose derived an inequality~\cite{PenroseNakedSing1973} relating the Arnowitt–Deser–Misner (ADM) mass $m_{\mathrm{ADM}}$, the mass of a gravitating system defined in terms of the spacelike asymptotic behavior \cite{Arnowitt:1960es}, to the area $A$ of any cross section of the event horizon. Specifically, the inequality states 
\be
\label{eqn:Penroseoriginal}
m_{\mathrm{ADM}} \geq \sqrt{\frac{A}{16 \pi}} \,.
\ee

In his original work, Penrose connects this inequality with the weak cosmic censorship conjecture in the following way: considering a collapse picture, at some point, a marginally trapped surface forms. Now, Penrose imposes the Null Convergence Condition (NCC), which requires that the Ricci tensor satisfies
\be
\label{eqn:NCC}
R_{\mu \nu}U^\mu U^\nu\geq0 \,,
\ee
for all null vector fields $U^\mu$ at every point in spacetime. With the use of the Einstein equation
\begin{equation}
\label{eqn:Einstein}
R_{\mu\nu} - \frac{1}{2} R\,{g}_{\mu\nu} = 8\pi T_{\mu\nu} \,,
\end{equation}
this condition translates to an energy condition about matter, the Null Energy Condition (NEC)\footnote{Sometimes in the literature the term is also used instead of the NCC \cite{Kontou:2020bta}.}, where the Ricci tensor is replaced by the energy-momentum tensor $T_{\mu \nu}$. The NCC is a crucial assumption in several foundational results in general relativity, such as Penrose's singularity theorem~\cite{penrose1965gravitational}. The theorem predicts null geodesic incompleteness using the NCC along with the requirement of a trapped surface and some causality assumptions. Thus the spacetime with the trapped surface is necessarily singular. If we assume weak cosmic censorship, an event horizon also forms outside the marginally trapped surface.

The NCC and the existence of an event horizon are the necessary conditions for Hawking's black hole area theorem~\cite{Hawking:1971vc}, which shows that the event horizon of a black hole is always non-decreasing. Finally, the initial Penrose formulation invokes the assumption that the black hole spacetime approaches the Kerr metric. This was an assumption, but theorems of uniqueness of Kerr strongly support it. More precisely, Robinson \cite{Robinson:1975bv} showed that the only stationary solution of the vacuum Einstein equations when a regular event horizon is present is the Kerr metric\footnote{Further works removed some of the assumptions of this theorem.}. Putting the area theorem and this assumption together gives
\be \label{eqn:Penarg_AreaTheorem}
A \leq A_{\text{Kerr}} \,,
\ee
where $A_{\text{Kerr}}$ is the area of the event horizon in a Kerr spacetime. We have 
\be
\label{eqn:Penarg}
A_{\text{Kerr}}=8\pi M_{\text{Kerr}}\left( M_{\text{Kerr}}+\sqrt{M^2_{\text{Kerr}}-L^2/M_{Kerr}^2} \right) \leq 16\pi M^2_{\text{Kerr}}= 16\pi m^2_{\text{B}} \,,
\ee
where $L$ is the total angular momentum and the last equality arises from the fact that the mass in a Kerr spacetime asymptotically approaches the Bondi mass $m_B$ of the evolving spacetime. The Bondi mass \cite{Bondi1962-rc} is similar to the ADM mass but it is computed at null infinity instead of the spatial infinity. It was conjectured by Penrose but proven later under certain assumptions \cite{MarcMars_2007} that $m_{\mathrm{ADM}} \geq m_{\text{B}}$. In particular, since gravitational waves carry positive energy, the Bondi mass cannot increase in the future. So this is true provided the Bondi mass approaches the ADM mass $m_{\mathrm{ADM}}$ of the initial slice, which is only known under additional assumptions. If instead, the original inequality \eqref{eqn:Penroseoriginal} fails then one can argue that weak cosmic censorship fails and then a naked singularity develops.

The second formulation is expressing the Penrose inequality \textit{locally}, meaning in terms of initial data (see Preliminaries \ref{sec:preliminaries}), i.e., maximal spacelike slices of the spacetime. The connection with weak cosmic censorship is similar, as if it holds, then an initial data set containing a future trapped surface must develop a singularity, and thus an event horizon. Here the NCC \eqref{eqn:NCC} is again necessary. The intersection of the event horizon and the initial data set defines a spacelike surface $H_S$ that
separates the black hole from the asymptotic region. This follows from the fact that if NCC is satisfied, the trapped surface $\Sigma$ is behind the event horizon \cite{Hawking:1971vc}. It is not necessarily true that $|\Sigma| \leq |H_S|$ but it is true that $|\Amin(\Sigma)| \leq|H_S|$ where $\Amin(\Sigma)$ is the infimum of the areas of all spacelike surfaces enclosing the trapped surface $\Sigma$. Then if the global Penrose inequality \eqref{eqn:Penroseoriginal} holds
\be
\label{eqn:Penroselocal}
m_{\mathrm{ADM}} \geq \sqrt{\frac{A_{\min}(\Sigma)}{16\pi}} \,.
\ee
The same argument holds for marginally trapped surfaces. Penrose's idea was that if initial data sets were found that violate Eq.~\eqref{eqn:Penroselocal}, then the weak cosmic censorship conjecture is violated - given that the other assumptions are satisfied. However, Eq.~\eqref{eqn:Penroselocal} can also be proven on its own using properties of initial data. Note that a proof of Eq.~\eqref{eqn:Penroselocal} is in support of weak cosmic censorship but does not fully prove it. For the initial data set version, instead of the NCC, a different energy condition, the geometric form of the dominant energy condition (DEC) is used
\be
\label{eqn:gDEC}
G_{\mu \nu} t^\mu \xi^\nu \geq 0 \,,
\ee
where $G_{\mu \nu}$ is the Einstein tensor and $t^\mu$ and $\xi^\nu$ co-oriented timelike vectors. The geometric form of the DEC implies the NCC \cite{Kontou:2020bta} as well as the timelike convergence condition (where $t^\mu$ and $\xi^\mu$ are the same vector field). In this context, the DEC can be expressed as a condition on the initial data set.  

Recently, significant progress has been made towards proving \eqref{eqn:Penroselocal} under certain conditions, such as for  maximal time-symmetric spacelike slices satisfying the geometric form of the weak energy condition, the timelike convergence condition. This is known as the Riemannian Penrose inequality, which also holds in the presence of multiple black holes \cite{HuiskenIlmanen2001,Bray2001,BrayLee2009,AgostinianiMantegazzaMazzieriOronzio2022}. However, the most general case, for maximal spacelike slices, remains open except for the spherically symmetric case \cite{Hayward1996,MalecMurchadha1994,Bray:2009au,BrydenKhuriSormani2020,Lee2019-cz}. For the general case, the inequality fails if the initial data set satisfies only the timelike convergence condition, as demonstrated in \cite{Jaracz2022}. A review of relevant proofs and open questions is given in Ref.~\cite{Mars:2009cj}.

One issue with all these results is that quantum fields are known to violate the usual pointwise energy conditions, including the NEC~\cite{Epstein1965}. For example, an evaporating black hole is known to violate all such energy conditions. From this, it is clear that the proposed inequality, as well as the previously mentioned classical theorems, fail at the semiclassical level when a field-theoretic treatment is introduced~\cite{BoussoMoghaddamTomasevic2019}.

Notably, semiclassical singularity theorems were recently proved \cite{FewsterKontou_2022, Freivogel:2020hiz}, where the pointwise energy conditions were replaced by averaged energy conditions. In these cases, those can be derived, with the use of the semiclassical Einstein equation
\be
\label{eqn:SEE}
R_{\mu\nu} - \frac{1}{2} R \,g_{\mu\nu} = 8 \pi \langle T_{\mu\nu} \rangle_\omega \,.
\ee
The semiclassical Einstein equation assumes that the quantum fields generate a classical curved background, valid as an approximation in a certain regime. A solution to the semiclassical Einstein equation consists of a metric $g$ and a state $\omega$ and it is generally difficult to find. In this approach, it is used to translate results about the expectation value of the stress-energy tensor given by Quantum Energy Inequalities (QEIs) to curved spacetimes. QEIs are derived directly by quantum field theory and bound the averaged negative energy of the expectation value of the stress-energy tensor~\cite{Ford1978}. 

The Hawking black hole area theorem~\cite{Hawking:1971vc} is different from the singularity theorems as it is always violated semiclassically. Nevertheless, it was obtained~\cite{Kontou:2023ntd} under a weaker condition than the NCC \footnote{Note that this condition is averaged but not obeyed by quantum fields in general.}. A weaker version of the theorem where the area of the black hole is allowed to decrease but with a bounded rate was also derived in Ref.~\cite{Kontou:2023ntd} using QEIs. 

Therefore, there is evidence that such theorems can admit some semiclassical extension with an appropriate energy condition. It is worth mentioning that another approach to this problem is to replace the area of surfaces which corresponds to the classical black hole entropy in the problem, with their generalized entropy (see \cite{BoussoMoghaddamTomasevic2019} and references within). 

Given the number of results in general relativity that can be recovered at the semiclassical level when a better-suited energy condition is satisfied, it is natural to wonder whether one could derive the Penrose inequality under weaker energy conditions. That is, even if the NCC is violated pointwise, it might still be possible to derive the same inequality or a modified version of it with an average version of it.

These are the two main results of this paper:
\begin{itemize}
    \item 
    A proof of the local version of the Penrose inequality, assuming spherical symmetry of the initial data set, using an average form of the DEC (Theorem~\ref{PI-SADEC}).  
    \item 
    A proof of a modified version of the Penrose inequality for evaporating black holes using a condition inspired by QEIs and an average form of the DEC (Theorem~\ref{the:PIevapo} and Corollary~\ref{cor:PIevapoqei}).
\end{itemize}

The paper is organized as follows: In Sec.~\ref{sec:preliminaries} we introduce the basics of the initial data formulation. In Sec.~\ref{sec:PIinitialdata} we present our first main result with a proof of the Penrose inequality with an averaged energy condition while in Sec.~\ref{sec:examples} we provide relevant examples of its use. In Sec.~\ref{sec:dynamical} we prove a theorem for the case of dynamical black holes including the evaporating case. In  Sec.~\ref{sec:applications} we apply the previous results using a condition inspired by QEIs and provide relevant examples. We conclude in Sec.~\ref{sec:discussion} with a summary and discussion of future work.

\begin{table}
    \centering
    \begin{tabular}{cc}
       ADM  & Arnowitt-Deser-Misner \\
       NEC  & Null Energy Condition \\
       NCC  & Null Convergence Condition\\
        DEC & Dominant Energy Condition \\
        QEI & Quantum Energy Inequality \\
        MOTS & Marginally Outer Trapped Surface \\
        MITS & Marginally Inner Trapped Surface \\
        $d\nu$-ADEI & $d\nu$ Averaged Dominant Energy Inequality \\
        $S$-DEI & Spherical Averaged Dominant Energy Inequality \\
        $S$-DEC & Spherical Averaged Dominant Energy Condition \\
        WEC & Weak Energy Condition \\
        SNEC & Smeared Null Energy Condition \\
    \end{tabular}
    \caption{Abbreviations used in the text (in order of appearance).}
    \label{tab:my_label}
\end{table}
\vspace{0.2in}

\noindent
\textbf{Conventions:} All manifolds are assumed to be smooth, Hausdorff, second-countable, and connected. Lorentzian manifolds $(\overline{M}^n, \overline{g})$ are denoted using an overline and are $n$-dimensional with $n \geq 4$, while Riemannian manifolds $(M^d, g)$ are $d$-dimensional. Lorentzian metrics have signature $(-, +, \ldots, +)$, and both $\overline{g}$ and $g$ are assumed to be smooth.

\section{Preliminaries}
\label{sec:preliminaries}

A spacetime is a pair $(\overline{M}^n, \overline{g})$, where $\overline{M}^n$ is a $n$-dimensional, $n\geq4$, manifold endowed with a Lorentzian metric $\overline{g}$. It is assumed that $(\overline{M}, \overline{g})$ is time-orientable.  A spacetime $(\overline{M}, \overline{g})$ satisfying the Einstein equation \eqref{eqn:Einstein} is said to obey the geometric form of the \textit{Dominant Energy Condition} (DEC) \eqref{eqn:gDEC} provided the energy-momentum tensor $T_{\mu\nu}$ satisfies $T_{\mu\nu} t^\mu \xi^\nu \geq 0 $ for all future pointing causal vectors $t^\mu, \xi^\nu$. 

Any embedded spacelike hypersurface $M$ in a spacetime gives rise to an \textit{initial data set} $(M^d, g, \mathcal{K})$, where $g$ is the Riemannian induced metric on $M$ and $\mathcal{K}$ the second fundamental form associated with the future directed timelike unit vector field ${t}$. The geometric form of the DEC holding in $\overline{M}$ guarantees that
\be
\label{eqn:datadec}
\mu \geq |J|_g \quad \text{along $M$},
\ee
where the scalar $\mu = {G_{\mu\nu}}{t}^\mu{t}^\nu$ and the one-form $J^\mu = G^{\mu\nu} t_\nu$. In particular, by the Gauss-Codazzi equations, one can express $\mu$ and $J$ solely in terms of the initial data set,
\be
\label{eqn:rhoids}
\mu = \frac{1}{2} \left(R_g - |\mathcal{K}|_g^2+(\operatorname{tr}_g\mathcal{K})^2\right)\,,
\ee
\be
\label{eqn:tids}
J =\operatorname{div}_g\mathcal{K} - d(tr_g\mathcal{K}) \,.
\ee
For any Riemannian manifold $(M, g)$ and smooth symmetric $(0, 2)$ tensor field $\mathcal{K}$, we can define $\mu$ and $J$ as the right-hand sides of the above equations. If the triple $(M, g, \mathcal{K})$ is an initial data set in a  spacetime satisfying the Einstein equations \eqref{eqn:Einstein}, $\mu/8 \pi$ and $J/8 \pi$ represent the physical matter fields and are called the \textit{energy density} and \textit{current density}, respectively.

The case $\mathcal{K}\equiv 0$ is referred to as time-symmetric (or Riemannian) and the DEC \eqref{eqn:datadec} reduces to $R_g \geq 0$. Until Sec.~\ref{sec:dynamical}, the discussion will focus exclusively on initial data sets, as the DEC can be formulated independently of an underlying spacetime.

An initial data set $(M^d, g, \mathcal{K})$ is said to be \textit{asymptotically flat} if there exists a compact set $K\subset M$ such that $M \setminus K$ is diffeomorphic to $\mathbb{R}^d \setminus \overline{B_{1}(0)}$, i.e., there exists a diffeomorphism 
    \be
    \Phi:  M \setminus K \rightarrow \mathbb{R}^d \setminus \overline{B_1(0)} \,,
    \ee
    where $\overline{B_1(0)}$ is the standard closed unit ball\footnote{For simplicity, we assume that the manifold has only one end $M \setminus K$.}. Moreover, if we think $\Phi$ as a coordinate chart with coordinates $x^1,\ldots, x^d$, then, in this coordinate chart, we assume that $g_{ij}$ satisfies\footnote{The notation $O_k(|x|^{q})$, for $k\geq1$, denotes a function in $C^k_{-q}$, where $f \in C^k_{-q}$ if $ |f(x)| + |x| |\partial f(x)| + \ldots + |x|^k |\partial^k f(x)| < C |x|^{-q}$ for some constant $C$.}
    \begin{align}
        g_{ij}(x) & = \delta_{ij} + O_2(|x|^{-q}) \label{eq:DecayMetric}\\
        \mathcal{K}_{ij}(x) & = O_1(|x|^{-q-1}) \,,\label{eq:Decay2form}
    \end{align}
for some $q>(d-2)/2$ and $\mu$ and $J$ are integrable over $M$. 

With these asymptotics, the \textit{ADM energy-momentum} $(E, P)$ is well-defined \cite{Bartnik1986, Chrusciel1986} and is given by 
\be
E:= \lim_{r\rightarrow \infty} \frac{1}{2(d-1)\omega_{d-1}} \int_{\mathbb{S}_r} (g_{ij,i} - g_{ii,j}) {n}^i d\mu \,,
\ee
and
\be
P_i:= \lim_{r\rightarrow \infty} \frac{1}{(d-1)\omega_{d-1}} \int_{\mathbb{S}_r} (\mathcal{K}_{ij} -  (\operatorname{tr}_g \mathcal{K})g_{ij}) {n}^j d\mu \,.
\ee
Here, $\mathbb{S}_r$ is the coordinate sphere of radius $r$, $n$ is its outward unit normal, $d\mu$ is the Euclidean area element of $\mathbb{S}_r$ and $\omega_{d-1}$ is the area of the standard unit $(d-1)$ sphere. The ADM mass of an asymptotically flat initial data set is defined to be
\be
m_{\mathrm{ADM}} = \sqrt{E^2 - \delta_{ij}P_iP_j} \,.
\ee

Let $\Sigma^{d-1} \subset M$ be a closed and two-sided hypersurface with outward unit normal vector $\nu$. An important quantity associated with $\Sigma$ are the {\textit outward/inward null expansions} defined as
\be\label{eq:MOTS_MITS}
\theta_\pm = \operatorname{tr}_\Sigma \mathcal{K} \pm H_\Sigma,
\ee
where $H_\Sigma=\operatorname{div}_\Sigma \nu$ is the mean curvature of $\Sigma$ associated with the outward normal vector field.

In a spacetime, the sign of the null expansions indicates the convergence and divergence of past and future directed null geodesics, as well as the gravitational field's strength near $\Sigma$. This can be clarified by the following argument. Let $t$ be the future directed timelike unit vector field orthogonal to $M$. Then $\ell_\pm= t \pm \nu$ are future directed null normal vector fields along $\Sigma$. Define the \textit{null second fundamental forms} $\chi_\pm (X,Y) = g(\nabla_X \ell_\pm, Y)$, for all $X,Y \in \mathfrak{X}(\Sigma)$, then the null expansions $\theta_\pm$ of $\Sigma$ are given by
$
\theta_\pm = \tr_\Sigma\chi_\pm = \tr_\Sigma \mathcal{K} \pm H_\Sigma,
$
and, in particular, $\theta_\pm = \operatorname{div}_\Sigma \ell_\pm$.

Outer and inner trapped surfaces are characterized by the inequalities $\theta_+<0$ and $\theta_-<0$, respectively, indicating that they may lie within a region of strong gravitational field. When $\theta_+=0$ ($\theta_-=0$), the surface $\Sigma$ is called a marginally outer (inner) trapped surface, or MOTS (MITS). Additionally, we refer to that $\Sigma$ as an outermost MOTS (MITS) if it is not enclosed by any other MOTS (MITS). In general, for stationary spacetimes, cross sections of the event horizon are MOTS and in dynamical spacetimes MOTS occurs inside the black hole region~\cite[Theorem 6.1]{Chrusciel2009-nt}.

\section{Local Penrose Inequality from average energy conditions}
\label{sec:PIinitialdata}

In this section, we introduce the notion of the averaged dominant energy inequality, which depends on a specific choice of measure. For spherically symmetric initial data sets, there is a natural choice of measure that allows us to establish a modified Penrose inequality and recover the classical Penrose inequality, even in the presence of pointwise violations of the DEC.

\subsection{Averaged dominant energy inequality}

Motivated by pointwise violations of the DEC, we propose an averaged condition that provides sufficient flexibility to be applied to a class of physically interesting examples. We note that this condition is not expected to be obeyed by quantum fields. The reason is that our condition is integrated over a spacelike hypersurface, and such conditions do not generally have a finite lower bound in quantum field theory~\cite{FordHelferRoman2002, Fewster:2002ne}. We should note however, that if the semiclassical Einstein equation is satisfied and the field generates an asymptotically flat hypersurface it seems that a lower bound should exist. As evidence, the condition is satisfied by various black hole geometries that violate the pointwise DEC, some motivated by quantum field theory corrections as we show in Sec.~\ref{sec:examples}. 

\begin{definition}[Averaged Dominant Energy Inequality]
    An initial data set $(M^d, g, \mathcal{K})$ satisfies the $d\nu$-\textit{Averaged Dominant Energy Inequality} ($d\nu$-ADEI) if there exists $K \in \mathbb{R}$ such that
    \be
    \label{eqn:dnadei}
    \int_{M} (\mu - |J|_g)\, d\nu \geq K \,,
    \ee
where $\mu$ and $J$ are defined in Eqs.~\eqref{eqn:rhoids} and \eqref{eqn:tids}, respectively, and $d\nu$ is a measure.
\end{definition}

A natural choice for $d\nu$ would be the induced measure $d\mu_g$, however this choice is not well-suited to address the class of the problems we aim to consider. Depending on the symmetries of the initial data or the structure of the proof techniques, more appropriate choices of measure $d\nu$ may rise. Let $(M^d, g, \mathcal{K})$ be a spherically symmetric initial data set with boundary, i.e., there is a diffeomorphism from $M^d$ to $[0, \infty) \times \mathbb{S}^{d-1}$ under which the metric takes the form
\be\label{eq:SphrMetricds}
g = ds^2 + r(s)^2d\Omega^2 \,,
\ee
for some smooth positive function $r(s)$, where $d\Omega^2$ is the standard metric on the sphere, while $\mathcal{K}$ can be written as
\be
\mathcal{K} = \mathcal{K}_{nn}(s) ds^2 + \frac{1}{d-1}\kappa(s) r(s)^2 d\Omega^2 \,,
\ee
where $\mathcal{K}_{nn}(s)$ and $\kappa(s)$ are smooth functions. Observe that $\kappa$ is the trace of $\mathcal{K}$ over the sphere at $s$ and $\mathcal{K}_{nn}=\mathcal{K}(n,n)$, where $n$ is the unit outward normal of the sphere at s.

In particular, the symmetric spheres are minimal at points where $dr/ds=0$. When this quantity is nonzero, the metric can be rewritten as
\be\label{eq:SphrSymMetric}
g = \frac{1}{V(r)} dr^2 + r^2d\Omega^2 \,,
\ee
where $V(r)$ is a smooth positive function, and  $d\Omega^2$ denotes the standard metric on the unit sphere $\mathbb{S}^{d-1}$. Given any $\mathbb{S}(r_0)$ for some $r_0>0$, let $B_{r_0}$ denote the region it bounds. Denote by $\omega_{d-1}$ the volume of the unit sphere $\mathbb{S}^{d-1}$, and choose the measure $d\nu = \omega_{d-1}^{-1}\sqrt{V(r)} d\mu_g$. A standard computation reveals that
    \be \label{eq:S-ADEI}
     \int_{M\setminus B_{r_0}} \left(\mu - |J|_g\right)\,d\nu = \frac{1}{\omega_{d-1}}\int_{r_0}^{\infty} \int_{\mathbb{S}(r)} (\mu  - |J|_g) \, d{S}dr.
    \ee
If the spherically symmetric initial data set has $d\nu$-ADEI bounded below by $K \in \mathbb{R}$ where 
\be
d\nu=\frac{1}{\omega_{d-1}}\sqrt{V(r)} d\mu_g,
\ee
with respect to $\mathbb{S}(r_0)$, then $(M^d, g, \mathcal{K})$ satisfies the \textit{Spherical Dominant Energy Inequality} (S-DEI) bounded below by $K$ with respect to $\mathbb{S}(r_0)$. The S-ADEI can also be read as
    \be
    \int_{r_0}^{\infty} r^{d-1} \left( \mu  - |J|_g \right)\,dr \geq K.
    \ee
In particular, when $K=0$, we refer to the S-DEI as the \textit{Spherical Averaged Dominant Energy Condition} (S-DEC) with respect to $\mathbb{S}(r_0)$.

\subsection{Modified Penrose inequality}

In this subsection, we apply the S-ADEI to show a modified Penrose inequality for complete spherically symmetric initial data sets and state some immediate consequences.

To formulate the S-ADEI for spherically symmetric initial data sets with boundary $\partial M$ given by an outermost MOTS, we must ensure the existence of a region free of minimal surfaces, allowing the metric to be written as \eqref{eq:SphrSymMetric}. The existence of an outermost MOTS $\partial M$ is not sufficient. However, the presence (or absence) of a MITS, and then an outermost MITS, guarantees that no minimal spheres exist outside the outermost MITS (or $\partial M$)\footnote{This fact is proven alongside the main theorem.}. Consequently, the S-ADEI is well-defined with respect to the outermost MITS (or $\partial M$).  

With this in mind, we can now state our main theorem.

\begin{theorem}
\label{PI-SADEC}
Let $(M^d, g, \mathcal{K})$, $d<8$, be a complete spherically symmetric, asymptotically flat initial data set such that $\partial M$ is an outermost MOTS. If the S-ADEI is bounded below by $K \in \mathbb{R}$ with respect to the outermost MITS (or $\partial M$ in its absence), i.e.,
\be
\label{eqn:SADEI}
\int_{r_0}^{\infty} r^{d-1} \left( \mu  - |J|_g \right)\,dr \geq K,
\ee
then
\be
\label{eqn:peninitial}
m_{\mathrm{ADM}} - \frac{K}{(d-1)} \geq \frac{1}{2}\left(\frac{\Amin(\partial M)}{\omega_{d-1}}\right)^{\frac{d-2}{d-1}}.
\ee
In particular, if there is no outermost MITS outside $\partial M$, then
\be
\label{eqn:peninitial2}
m_{\mathrm{ADM}} - \frac{K}{(d-1)} \geq \frac{1}{2}\left(\frac{|\partial M|}{\omega_{d-1}}\right)^{\frac{d-2}{d-1}}.
\ee
\end{theorem}
 
An immediate consequence of this result is that if the initial data set satisfies the S-DEC ($K \equiv 0$), then the classical Penrose inequality for spherically symmetric initial data sets is recovered. The key advantage of this approach is that it permits pointwise violations of the DEC. An interesting observation is that if a MITS, and thus an outermost MITS, exists outside $\partial M$, any pointwise violation of the DEC occurring between the outermost MITS and $\partial M$ is irrelevant to the proof, provided the S-ADEI holds with respect to the outermost MITS. This phenomenon can be interpreted as a shielding effect.

Whenever $|J|_g = 0$, which holds for time-symmetric initial data sets, we can also conclude that the Penrose inequality is reversed when the S-DEC is violated.
\begin{corollary}
\label{cor:reversed}
Let $(M^d, g, \mathcal{K})$, $d<8$, be a complete spherically symmetric, asymptotically flat initial data set such that $\partial M$ is an outermost MOTS. Assume that $|J|_g = 0$. The S-DEC with respect to the outermost MITS (or $\partial M$ in its absence) is violated, i.e.,
\be
\int_{r_0}^{\infty} r^{d-1} \mu  \,dr < 0,
\ee
if and only if the Penrose inequality is reversed,
\be
\label{eqn:reverse}
m_{\mathrm{ADM}} < \frac{1}{2}\left(\frac{{\Amin(\partial M)}}{\omega_{d-1}}\right)^{\frac{d-2}{d-1}}.
\ee
\end{corollary}

For time-symmetric initial data sets, the S-DEC reduces to a condition on the scalar curvature of the initial data set, given as an immediate corollary the following spherically symmetric Riemannian Penrose modified inequality.

\begin{corollary}
Let $(M^d, g)$, $d<8$, be a complete spherically symmetric, asymptotically flat manifold. If $\partial M$ is the outermost minimal surface and $(M, g)$ has S-ADEI bounded below by $K$ with respect to $\partial M$,  i.e.,
\be
\frac{1}{2}\int_{r_0}^{\infty} r^{d-1} R_g \,dr \geq K
\ee
then
\be
m_{\mathrm{ADM}} - \frac{K}{(d-1)} \geq \frac{1}{2}\left(\frac{|\partial M|}{\omega_{d-1}}\right)^{\frac{d-2}{d-1}}.
\ee  
\end{corollary}   

It is well known that the Penrose inequality implies the positive energy theorem (the $m_{\mathrm{ADM}}$ is always positive following the definition). Instead of assuming the DEC, if we impose an appropriate negative lower bound for the S-ADEI, the ADM energy remains positive, as shown in the following corollary.

\begin{corollary}\label{cor:lowerboundK}
Let $(M^d, g, \mathcal{K})$, $d<8$, be a complete spherically symmetric, asymptotically flat initial data set such that $\partial M$ is an outermost MOTS. If the S-ADEI with respect to the outermost MITS (or $\partial M$ in its absence) is bounded below as
\be
K \geq -\frac{(d-1)}{2}\left(\frac{{\Amin(\partial M)}}{\omega_{d-1}}\right)^{\frac{d-2}{d-1}},
\ee
then $E_{\mathrm{ADM}}\geq 0$.
\end{corollary}

Analogous to the Riemannian Penrose modified inequality, if the initial data set is time-symmetric, the S-DEC becomes an integral condition on the scalar curvature, and we recover a spherically symmetric version of the spherically symmetric Riemannian positive mass theorem that allows pointwise negative scalar curvature.

In what follows, we prove \Cref{PI-SADEC}, the remaining results of this section are particular cases of this result.

\begin{proof}[Proof of \Cref{PI-SADEC}]
This proof follows the approach in \cite[Theorem 7.46]{Lee2019-cz} and \cite{MarcMars_2007}, except for the argument involving the monotonicity of the local mass. The original proof is attributed to Malec and \'O Murchadha \cite{MalecMurchadha1994} as well as Hayward \cite{Hayward1996}. Since all other arguments remain the same or similar, we only present the main ingredients and show where we apply the S-ADEI to conclude the result. 

We need to separate in two cases (i) there is an outermost MITS outside $\partial M$ or (ii) not. From now on, assume case (ii), i.e., there is no outermost MITS outside $\partial M$.

Since the initial data set is spherically symmetric, the metric can be written in the form of \eqref{eq:SphrMetricds}. The obstruction to writing the metric in form of \eqref{eq:SphrSymMetric} is given by the existence of minimal spheres, i.e., $dr/ds=0$. Suppose that there is a minimal sphere outside $\partial M$ and denote by $\kappa$ the trace of $\mathcal{K}$ over the symmetric sphere of radius $r$. The decay rates for $g$ and $\mathcal{K}$ in \eqref{eq:DecayMetric} and \eqref{eq:Decay2form} imply $H=O(r^{-1})$ and $\mathcal{K}_{ij}=O(r^{-q-1})$, for $q>(d-2)/2$, so $H>|\kappa|$ for large $r$. Since $H=0$ at the minimal sphere, it follows by continuity that $H=|\kappa|$ on some sphere enclosing $\partial M$. Then, this sphere would be a MOTS or a MITS (see \eqref{eq:MOTS_MITS}). Because $\partial M$ is the outermost MOTS, it must be a MITS, which implies the existence of an outermost MITS (only for $d<8$) enclosing it~\cite{Andersson:2007gy,Eichmair2009, Andersson2011}. However, by hypothesis, there is no outermost MITS outside $\partial M$. Therefore, this argument rules out the existence of a minimal sphere outside $\partial M$.

Then, we can write the metric as
\be \label{eq:metric_sphsym}
g = V^{-1} dr^2 + r^2d\Omega^2 \,,
\ee
on $[r_0, \infty) \times \mathbb{S}^{d-1}$, where $\sqrt{V}=\frac{dr}{ds}$ is a smooth function of $r$, and $r_0$ satisfies $|\partial M| = \omega_{d-1} r_0^{d-1}$, where $r_0$ is the radius of $\partial M$. 

We consider the local mass $m(r)$ given by
\begin{align}
m(r)&= \frac{1}{2} r^{d-2}(1-V) + \frac{1}{2(d-1)^2} \kappa^2 r^d.
\end{align}
The first derivative of $m(r)$ is given by
\be
m^\prime(r) = \frac{1}{d-1} r^{d-1} \left(\mu - \frac{\kappa}{H} J(\nu) \right).
\ee
where $\nu$ is the outward normal vector field and $H$ the mean curvature of the sphere of radius $r$ (see \cite[Theorem 7.46]{Lee2019-cz} for details).

Since $H > |\kappa|$ for $r$ sufficiently large, if $H < |\kappa|$ for any sphere, by continuity $H=|\kappa|$, which leads to the existence of a MOTS or MITS enclosing $\partial M$. Therefore the outermost property of $\partial M$ and the absence of MITS outside $\partial M$ implies that we must have $H>|\kappa|$ in the interior of $M$. 

Assuming the S-ADEI \eqref{eqn:SADEI} with respect to $\partial M$, we have
\begin{align} \label{eq:ApplyS-ADEI}
m(r)\Big|^{\infty}_{r_0} = \int_{r_0}^{\infty} m^\prime(r)\,dr & =  \frac{1}{(d-1)}\int_{r_0}^{\infty} r^{d-1} (\mu - \frac{\kappa}{H} J(\nu)) \,dr  \nonumber \\
& \geq  \frac{1}{(d-1)} \int_{r_0}^{\infty} r^{d-1} \left(\mu -  |J|_g\right) \,dr  \geq \frac{K}{(d-1)} \,.
\end{align}
It is easy to check that since $\kappa^2 = H^2 = (d-1)^2 V(r_0) r_0^{d-1}$ at $\partial M$, we have
\begin{align}
\frac{1}{2}\left(\frac{|\partial M|}{\omega_{d-1}}\right)^{\frac{d-2}{d-1}}=\frac{1}{2} r_0^{d-2} =m\left(r_0\right)
\leq m(\infty)  - \frac{K}{(d-1)},
\end{align}
where we have applied the S-ADEI with respect to $\partial M$. Since $m(\infty) = E_{\mathrm{ADM}}$, we conclude
\begin{align}
\frac{1}{2}\left(\frac{|\partial M|}{\omega_{d-1}}\right)^{\frac{d-2}{d-1}} \leq E_{\mathrm{ADM}}  - \frac{K}{(d-1)} \,.
\end{align}
In particular, spherical symmetry implies that the ADM momentum is zero, then $|E_{\mathrm{ADM}}|$ agrees with $m_{\mathrm{ADM}}$.

Returning to case (i), where an outermost MITS $S$ exists, this also ensures that no minimal spheres exist and that $H > |\kappa|$ outside $S$, as the presence of such a surface would imply the existence of another MITS, contradicting the outermost property of $S$. Let $r_0$ be the radius of $S$, then \eqref{eq:ApplyS-ADEI} holds and $\kappa^2 = H^2 = (d-1)^2 V(r_0) r_0^{d-1}$ at $S$. Applying the analogous argument to case (i) and that $\Amin(\partial M) \leq |S|$, we have
\begin{align}
\label{eq:MITSargument}
\frac{1}{2}\left(\frac{\Amin(\partial M)}{\omega_{d-1}}\right)^{\frac{d-2}{d-1}} \leq \frac{1}{2}\left(\frac{|S|}{\omega_{d-1}}\right)^{\frac{d-2}{d-1}}
\leq m_{\mathrm{ADM}}  - \frac{K}{(d-1)},
\end{align}
where we have applied the S-ADEI with respect to $S$ and that $m(\infty) = E_{\mathrm{ADM}}$ and $|E_{\mathrm{ADM}}| = m_{\mathrm{ADM}}$ for spherically symmetric initial data sets.
\end{proof}
The proof of \Cref{cor:reversed} is a direct consequence of $|J|_g=0$ and \eqref{eq:ApplyS-ADEI}, since the first inequality becomes an equality, and the violation of the Penrose inequality follows directly from the violation of the energy condition, and vice versa.

\section{Examples and counterexamples}
\label{sec:examples}

This section is devoted to discussing some examples that violate the pointwise DEC. In certain cases, we can show that despite the pointwise violation, the S-DEC or the S-ADEI can still be satisfied and we can conclude the classical Penrose inequality or its modified version. First, we present examples that satisfy the S-DEC. Then we briefly discuss some counterexamples to the S-DEC and hence the classical Penrose inequality. Finally we present some non-spherically symmetric candidates that nevertheless can obey a form of the $d\nu$-ADEI with the appropriate measure.

\subsection{Effective Quantum Potential}

Casadio \cite{Casadio2022} investigated a quantum description of black holes based on coherent states of gravitons sourced by a matter core. The metric is similar to the Schwarzschild metric, but the Newtonian potential is replaced by an effective quantum potential $V_{\mathrm{QN}}$. The line element takes the form
\be
\mathrm{d}s^2 =-\left( 1+2 \, V_{\mathrm{Q N}} \right)d t^{2}+\frac{\mathrm{d} r^{2}} {1+2 \, V_{\mathrm{Q N}}}+r^{2} \, d \Omega^{2}.
\ee
In Ref.~\cite{Casadio2022} the author considers quantum coherent states that approximate the Newtonian potential at $r \gtrapprox R_H=2G_N m_{\mathrm{ADM}}$, where $G_N$ is the Newton constant. This is achieved by building the coherent states with modes of wavelength larger than some fraction of the size of the gravitational radius $R_H$. In turn, the approach yields an effective quantum potential $V_{\mathrm{QN}}$
\be
V_{\mathrm{QN}} \simeq -\frac{2 G_N M}{\pi r} \int_{0}^{r / R_{\mathrm{s}}}d z \, \frac{\operatorname{s i n} z} {z} \;, 
\ee
where $M$ is the mass of the matter core, and $R_s$ can be interpreted as the finite radius of a would-be-regular matter source. The classical limit is retrieved by setting $R_s \to 0$.

The horizon radius $r = r_H$ is given by the solution of $2 V_{\mathrm{QN}} = -1$, which can be computed numerically for different values of $R_s$. The quantum corrected metric still contains a horizon at $r_H \approx  2 G_N M$ for $R_s$ small.

Since this is a static spacetime, the natural initial data sets are the time-constant slices which are time-symmetric ($\mathcal{K}\equiv 0$). Thus, we only need to compute the scalar curvature to verify whether the S-DEC is satisfied. The scalar curvature $R$ for a time-constant slice is given by
\be
R \simeq \frac{8 G_N M \sin(\frac{r}{R_s})}{\pi r^3} \,,
\ee
which clearly violates the pointwise energy condition DEC. Nevertheless, the scalar curvature is not negative everywhere and exhibits oscillating behavior. Interestingly, for certain values of $R_s$, the classical Penrose inequality remains satisfied whenever the S-DEC is satisfied.

In order to apply the S-DEC, we must determine the horizon position $r_H$ corresponding to each $R_s$ by solving the equation $2V_{QN}=-1$. Setting $G_N = M = 1$, the resulting values of $r_H$ are presented in Fig.~\ref{fig:OscScal}(b). It is observed that $r_H$ oscillates around $2$ as a function of $R_s$, due to the dependence of $V_{QN}$ on $R_s$. The Penrose inequality is satisfied when $r_H\geq 2$, and it is violated otherwise.

Despite the Penrose inequality being satisfied for some values of $R_s$, the scalar curvature exhibits an interesting oscillatory behavior and can be pointwise negative. To illustrate this statement, Fig.~\ref{fig:OscScal}(a) shows the pointwise values of $r^2(\mu-|J|_g)$, which is equals to $\frac{1}{2}r^2 R$, for different values of $R_s$. In particular, for $R_s = 0.18$, the S-ADEI is saturated, i.e., $K\approx 0$, slightly greater values of $R_s$ lead to a negative upper bound for the S-ADEI, resulting in a violation of the classical Penrose inequality. In fact, the inequality is reversed, as shown in Corollary~\ref{cor:reversed}. In contrast, for $R_s$ slightly smaller than $0.18$, both the S-DEC and the classical Penrose inequality are satisfied. This behavior is observed for several other values of $R_s<0.18$. Thus, the Penrose inequality holds even though the scalar curvature follows an oscillatory pattern.

\begin{figure}[!h]
    \centering
    \begin{minipage}[t]{0.48\linewidth}
        \centering        \includegraphics[height=6cm]{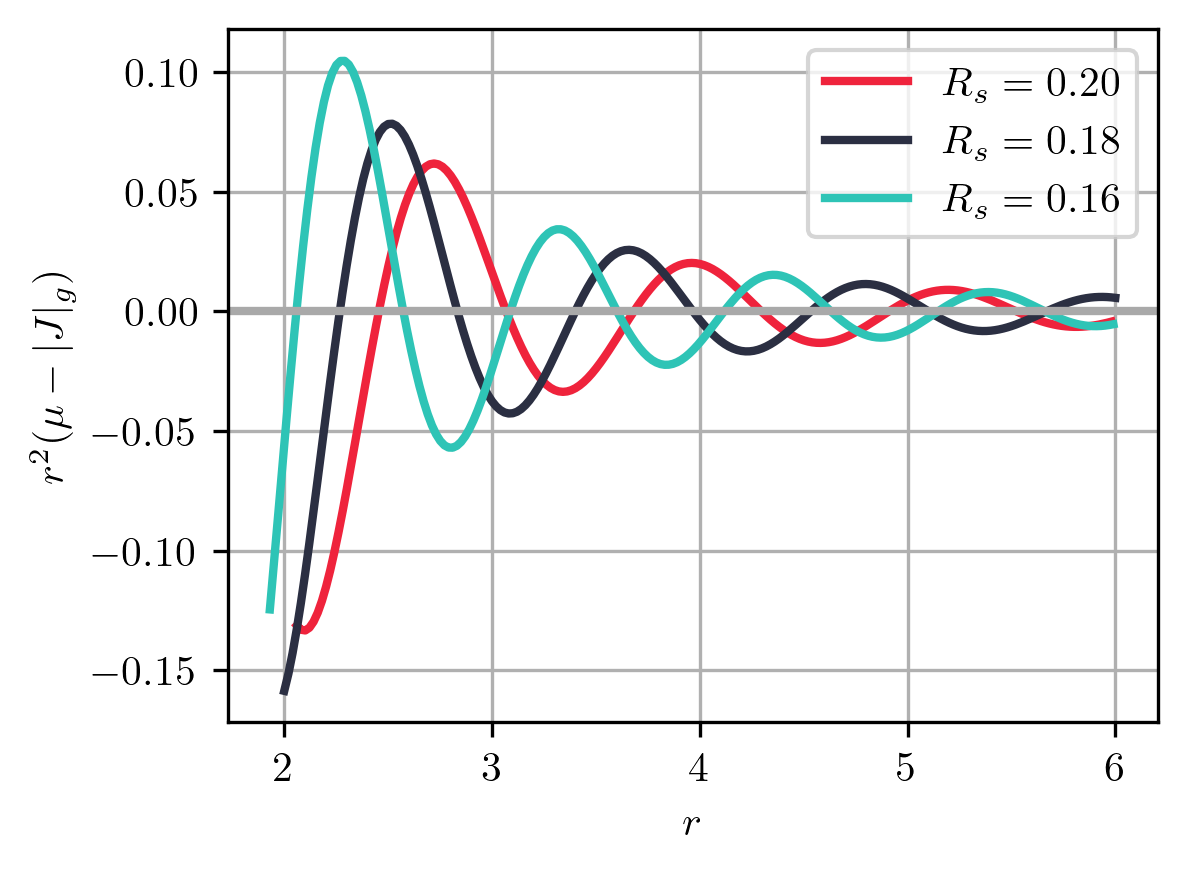}
        \caption*{(a)}
        \label{fig:OscScalA}
    \end{minipage}
    \hfill
    \begin{minipage}[t]{0.48\linewidth}
        \centering       \includegraphics[height=6cm]{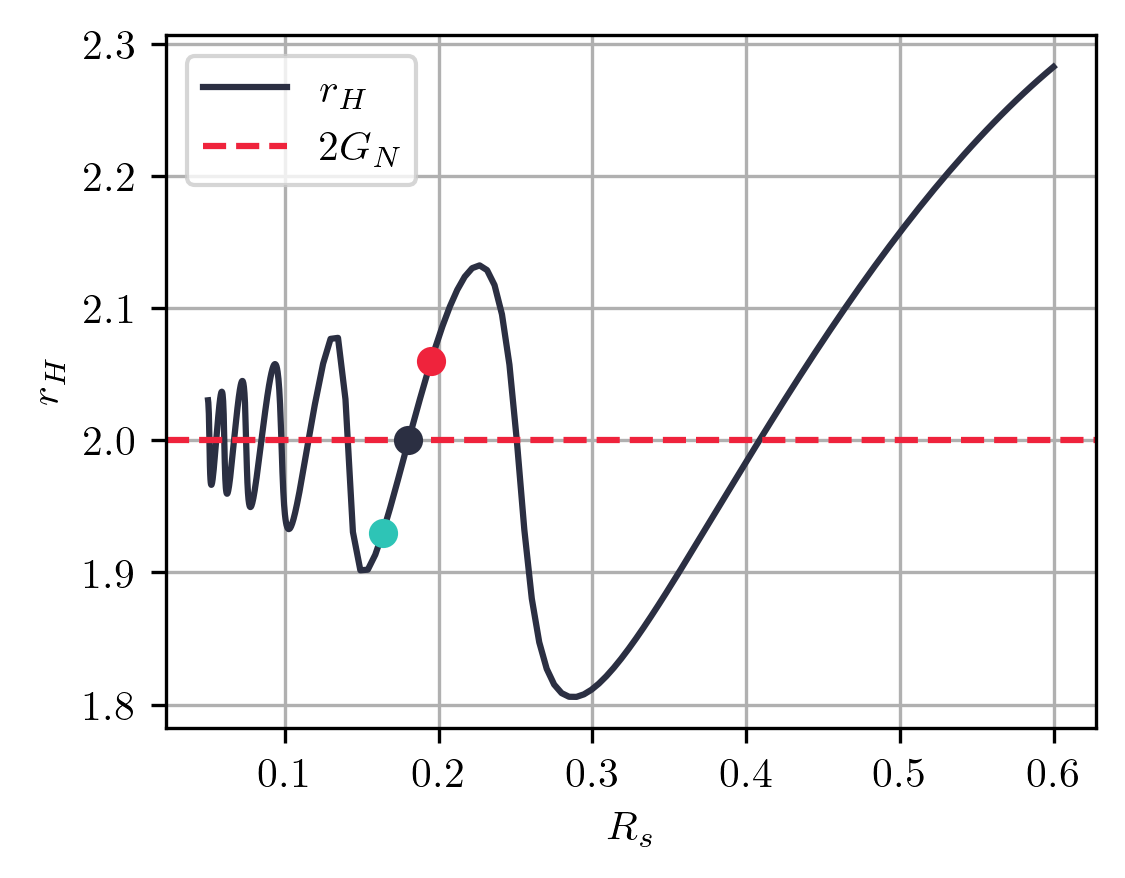}
        \caption*{(b)}
        \label{fig:OscScalB}
    \end{minipage}
    \caption{(a) The integrand of S-DEC for quantum black holes with $G_N = M = 1$. (b) Radius of the horizon $r_H$ for different values of $R_s$. The dots correspond to the lines in plot (a).}
    \label{fig:OscScal}
\end{figure}

\subsection{Third-order perturbation Schwarzschild}

Visser \cite{visser1996} studied conformally coupled massless scalar fields in Schwarzschild spacetime using the Hartle–Hawking vacuum. Through approximations of the stress-energy tensor, he showed that various pointwise energy conditions are violated. In particular, the DEC is violated for $r\lesssim3M$, and the NEC is violated for $r\lesssim2.3M$. However, for $r\gtrsim3M$, these pointwise energy conditions are satisfied again. Unfortunately, no exact four-dimensional metric is available, only approximations near the horizon exist, which are insufficient for our analysis. To circumvent this issue and reproduce these features, we focus on a class of examples represented by a third-order perturbation of Schwarzschild spacetime, parameterized by constants $A$ and $B$, which exhibits similar behavior near the horizon.

We propose the following family of metrics 
\be
\overline{g}= -f(r)d t^{2} + \frac{1}{f(r)}dr^{2}+r^{2}d\Omega^{2} \,,
\ee
where
\be
f(r) = \left(1-\frac{2M}{r}\right)+M\left(\frac{A}{r^2}-\frac{B}{r^3}\right) \,,
\ee
parameterized by constants $A$ and $B$. The horizon location $r_H$ can be determined by the roots of $f(r) = 0$. We constraint our analysis to values of $A$ and $B$ such that $f(r)$ has only one real positive solution.

Again, since this is a static spacetime, we only consider time-constant slices. Therefore, the scalar curvature for a time-constant slice is given by
\be
R = 2M \frac{(A r - 2B)}{r^5} \,.
\ee
Note that if $B=0$, the scalar curvature is always positive or negative, depending on the sign of $A$, and the metric looks like a charged Schwarzschild metric for positive $A$. The pointwise values of the function $r^2(\mu-|J|_g)$, with fixed $B = M = 1$ and varying values of $A$, are shown in Fig.~\ref{fig:ThirdOrdSch}.

\begin{figure}[!h]
    \centering
    \includegraphics[width=0.5\linewidth]{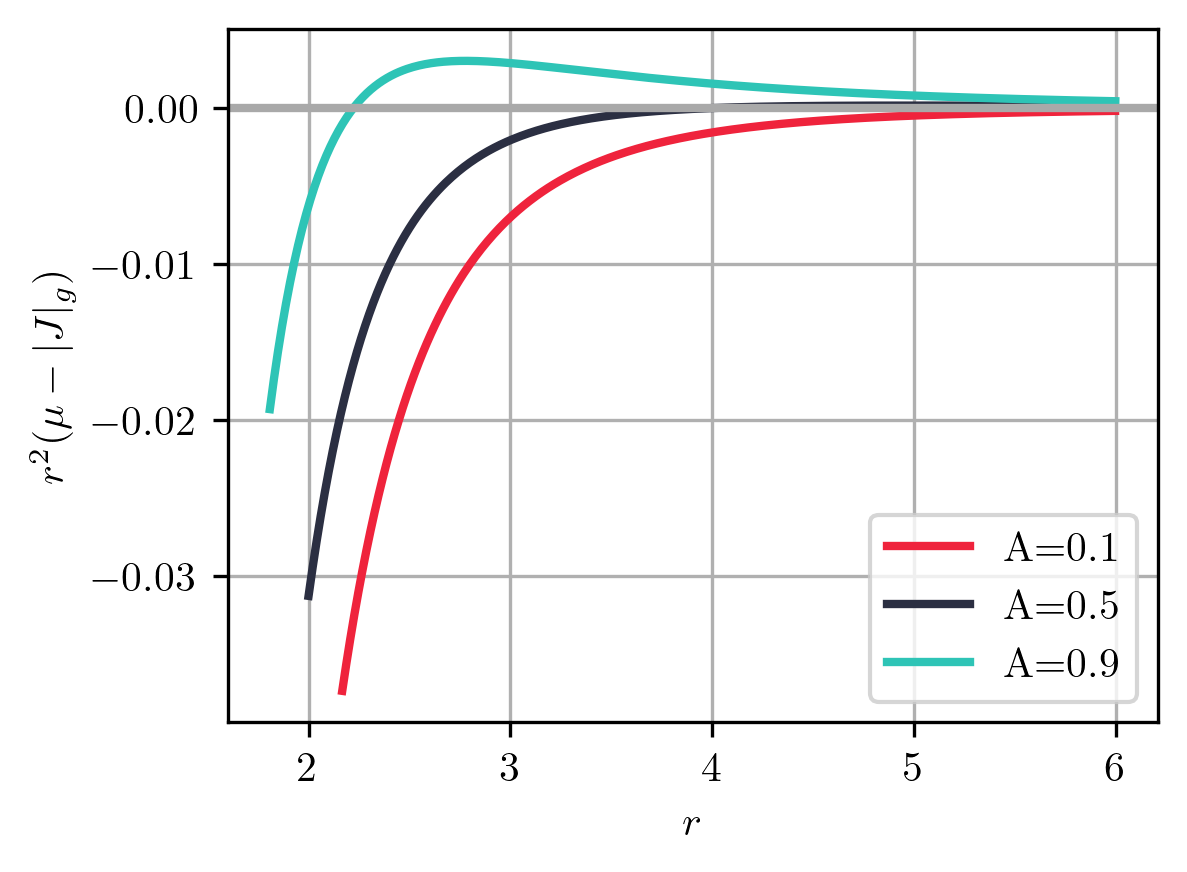}
    \caption{The integrand of S-DEC for the third-order modified Schwarzschild solution with $B = M = 1$ and varying $A$ demonstrates different behaviors.}
    \label{fig:ThirdOrdSch}
\end{figure}
Despite the presence of negative scalar curvature near the horizon, the classical Penrose inequality remains satisfied for both $A = 0.9$ and $A = 0.5$. However, for any values $A<0.5$, such as $A = 0.1$ shown in Fig.~\ref{fig:ThirdOrdSch}, the S-DEC is bounded both below and above by a negative constant, leading to a violation of the classical Penrose inequality and the reverse applies according to Corollary~\ref{cor:reversed}. Therefore, the classical Penrose inequality can still be satisfied even if the pointwise energy conditions are violated near the horizon, provided S-DEC holds.

\subsection{Vacuum Polarization}

Quantum-corrected metrics of the Schwarzschild solution from vacuum polarization without an event horizon are also found in the literature \cite{NavarroQC2023}. The authors obtained a static, spherically symmetric solution to the four-dimensional semiclassical Einstein equations, using the quantum vacuum polarization of a conformal field as a source. The metric is given by
\be
\overline{g}=-e^{-2 \phi(r)} d t^2+\frac{dr^2}{1-\frac{2m(r)}{r}}+r^2 d \Omega^2 \,,
\ee
where $\phi(r)$ and $m(r)$ are smooth functions that incorporate the vacuum polarization effects (see Eq. (18) and (19) \cite{NavarroQC2023}).

This metric has a coordinate singularity at $r_0$, the solution of $\overline{g}^{rr}(r_0)=0$ imposing $2m(r_0)=r_0$. The value of $r_0$ is close to the Schwarzschild coordinate singularity at $2M$, but is shifted by a positive amount that depends on $\hbar$\footnote{The units in Ref.~\cite{NavarroQC2023} are the geometrized units  $G_N=c=1$ to keep $\hbar$ explicit}. It is given by 
\be
r_0 = 2M + \frac{\sqrt{\hbar}}{4 \sqrt{70 \pi}} + \mathcal{O}(\hbar).
\ee
At $r_0$, unlike in the Schwarzschild case, there is no horizon since the metric component $\overline{g}_{tt}$ does not vanish at this point, meaning the static spacetime does not contain a horizon and thus does not describe a black hole. This fact is due to the backreaction effect that the quantum vacuum may produce in the metric. Despite the lack of a horizon, initial data sets from this spacetime are geometrically interesting. The coordinate singularity $r_0$ corresponds to a wormhole throat\footnote{\samepage Since this is a traversable wormhole, at least for light rays, it must violate the average null energy condition (ANEC) \cite{Graham:2007va}. The authors of \cite{NavarroQC2023} do not comment on that fact but from their construction it seems that there are complete null geodesics through the wormhole throat. This is an example where the S-DEC and the ANEC are violated. A similar construction is shown in \cite{Visser2003} where traversable wormholes are considered. The authors construct examples (see Specialization 1 \cite{Visser2003}) such that the $t$ constant slices ($\mathcal{K}\equiv0$) satisfy $\mu=0$ but the ANEC is violated. From this, it is clear that their examples satisfy the S-DEC while still violating the ANEC. These examples confirm that the ANEC and the S-DEC are independent conditions.} and is a minimal surface. Therefore, this surface is still suitable for our result and can be used to relate its area to the ADM mass.

The constant $t$-slices are time-symmetric, then S-ADEI depends only on the scalar curvature of the $t$-slice. In particular, a standard computation reveals that the scalar curvature is pointwise negative everywhere, so it violates the S-DEC and then the classical Penrose inequality by \Cref{cor:reversed}. Nevertheless, the following inequality holds
\be
m_{\mathrm{ADM}} + \OO(\sqrt{\hbar})  \gtrapprox \sqrt{\frac{|\partial M|}{16 \pi}} \,,
\ee  
where $\OO(\sqrt{\hbar})$ can be computed by explicitly using \eqref{eq:S-ADEI} and $\partial M$ is the surface defined by the coordinate singularity at $r_0$. Although the S-DEC does not hold, a form of the S-ADEI with $K$ depending on $\sqrt{\hbar}$ still holds, leading to a modified Penrose inequality.

In a subsequent paper \cite{navarro2024black}, the horizon was recovered by assuming a leading-order approximation, where the classical stress-energy tensor of the fundamental fields is traceless. The resulting metric takes the form of a charged Schwarzschild black hole, thereby satisfying the usual pointwise energy conditions.

Since the scalar curvature is negative pointwise, this example also violates the Weak Energy Condition (WEC), which is a weaker condition than the DEC, and requires $\mu \geq 0$ rather than $\mu\geq |J|_g$. Other physically relevant examples with quantum effects and violating the pointwise WEC in black hole spacetimes were proposed in Refs.~\cite{BoussoMoghaddamTomasevic2019, BoussoMogTom2}. These aim to serve as counterexamples in the same spirit as black hole evaporation is for Hawking’s area theorem, i.e., exhibiting a state where the null energy condition fails and so does the theorem conclusion. The construction considers a Boulware-like state of a massless scalar field in Schwarzschild spacetime and a Hartle–Hawking-like state near the horizon. The resulting configuration introduces negative energy outside the black hole on a time-symmetric initial slice, leading to a violation of the WEC. By carefully combining these vacua, this amount of negative energy is sufficient to construct an initial data set that violates the classical Penrose inequality not in order $O(\hbar)$ but in order $O(1)$. From the construction, it is possible to verify that the S-DEC is violated and as the initial data set is time-symmetric, $J\equiv0$, we can apply \Cref{cor:reversed} to conclude that a reverse Penrose inequality holds. It is unclear if S-ADEI is satisfied in this case.

\subsection{Violations of the Penrose inequality under the weak energy condition}

One approach to proving the Penrose inequality involves assuming the WEC ($\mu \geq 0$) instead of the DEC ($\mu \geq |J|_g$). In this section, we will see that counterexamples also exist within this class. However, it is possible to formulate a modified Penrose inequality by appropriately estimating $K$ based on the energy condition.

Recently, Ref.~\cite{Jaracz2022} presented a spherically symmetric initial data set with metric of the form \eqref{eq:SphrSymMetric} and $\partial M$ being an outermost MOTS that satisfies the WEC but violates the Penrose inequality. The core idea presented by the author involves prescribing scalar curvature on a spherically symmetric manifold and defining an appropriate $\mathcal{K}$ to ensure the WEC. In particular, this construction yields $\mu = 0$ and a nonzero one-form $J$, where $|J|_g \geq 0$ and $|J|_g\neq 0$. 

The argument relies on the Hawking mass and  inverse mean curvature flow. The Hawking mass is defined as 
\be
m_H(S)=\sqrt{\frac{|S|}{16 \pi}}\left(1-\frac{1}{16 \pi} \int_S H^2 d S\right) \,,
\ee
where $H$ is the mean curvature of the closed dimension-$2$ hypersurface $S$. It is a well-established result that if $S_t$ denotes a family of surfaces flowing out to infinity \cite{HuiskenIlmanen2001}, with each surface homotopic to the 2-sphere, then
\be
\lim_{t \rightarrow \infty} m_H(S_t) = E_{\mathrm{ADM}}
\ee
The construction in \cite{Jaracz2022} yields a sequence of pairs $(g_n, \mathcal{K}_n)$ satisfying 
\be \label{eq:Energy-Hmass}
E_{\mathrm{ADM}} - m_H(\mathbb{S}_1) = \lim_{t\rightarrow \infty} m_H(S_t) - m_H(\mathbb{S}_1) = -\frac{3}{4}
\ee
where $S_t$ is a family of surfaces obtained by flowing a coordinate sphere $r=1$ under the inverse mean curvature flow (see proof of Theorem 1.1 \cite{Jaracz2022}). The unit sphere $\mathbb{S}_1$ forms the boundary of the initial data set, though it is not necessarily an outermost MOTS. The outermost minimal area enclosure of the outermost MOTS $\Sigma$ (depending on $n$) has radius $r\geq 1$ and its minimal area enclosure must agree with $\Sigma$ by spherical symmetry \eqref{eq:SphrSymMetric} and uniqueness. This gives the following chain of inequalities 
\be \label{eq:AminArg}
\sqrt{\frac{A_{\min} (\Sigma)}{16 \pi}}=\sqrt{\frac{|\Sigma|}{16 \pi}} \geq \sqrt{\frac{\left|\mathbb{S}_1\right|}{16 \pi}}=\frac{1}{2}.
\ee
As $n \rightarrow \infty$, it is possible to show that $E_{\mathrm{ADM}}\rightarrow -\frac{1}{4}$, and then, for $n$ sufficiently large, we have
\be
\sqrt{\frac{A_{\min} (\Sigma)}{16 \pi}} \geq \frac{1}{2} > |E_{\mathrm{ADM}}| = m_{\mathrm{ADM}}.
\ee
Interestingly, Eq.~\eqref{eq:Energy-Hmass} can be connected with the S-ADEI and gives an explicitly bound. By Eq.~\eqref{eq:ApplyS-ADEI}, we have
\begin{align}
\frac{1}{2}K \leq \frac{1}{2}\int_{r_0}^\infty r^2(\mu-|J|_g)dr & \leq E_{\mathrm{ADM}} - \sqrt{\frac{|\Sigma|}{16 \pi}}\nonumber \\
& \leq  m_{\mathrm{ADM}} - \sqrt{\frac{|\mathbb{S}_1|}{16 \pi}} \nonumber\\ 
& = \lim_{t\rightarrow \infty}m_H(S_t) - m_H(\mathbb{S}_1) = -\frac{3}{4},
\end{align}
as a result, S-ADEI must have $K\leq-3/4$. Hence, an appropriate correction must be incorporated to preserve the inequality, as indicated by \Cref{PI-SADEC}. This adjustment enables formulating a modified Penrose inequality characterized by negative $K$.

\subsection{Non-spherically symmetric candidates}
\label{sec:NonSphSymCand}

While Theorem~\ref{PI-SADEC} applies to spherically symmetric initial data, there are some specific cases where while spherical symmetry is not present, if a version of the $d\nu$-ADEI \eqref{eqn:dnadei} is assumed, we can recover a modified Penrose inequality. 

\subsubsection{Graphical Case}
Lam \cite{lam2010} showed that an end of a time-constant slice of the 3-dimensional Schwarzschild metric can be isometrically embedded in $\mathbb{R}^4$ as the graph of a function over $\mathbb{R}^3 \setminus B_{2M}(0)$. Inspired by this example, he proposed a Penrose inequality for Riemannian manifolds $(M^d, g) = (\mathbb{R}^d, \delta + df \otimes df)$ given as the graph of a function $f:\mathbb{R}^d\setminus \Omega \rightarrow \mathbb{R}$, where $\Omega$ is an open set, satisfying certain conditions (see \Cref{theo:graph_theo}). Interestingly, all that is required to show the Penrose inequality is the non-negativity of
\be
\int_{M^d} R \frac{1}{\sqrt{1+|\nabla f|^2}} dV_g \,,
\ee
where $dV_g$ is the volume element induced by the metric. In terms of the integral condition, this shows that we can conclude a modified Penrose inequality if one assumes the $d\nu$-ADEI \eqref{eqn:dnadei} for the graphical manifold with 
\be \label{eq:Graphnu}
d\nu=\frac{1}{\omega_{d-1}\sqrt{1+|\nabla f|^2}}dV_g = \frac{1}{\omega_{d-1}}dV_\delta \,,
\ee
where $dV_\delta$ is the volume form of $\mathbb{R}^d$. Employing this measure, the result can be stated as follows. 

\begin{theorem}\label{theo:graph_theo}
 Let $\Omega$ be a bounded and open (but not necessarily connected) set in $\mathbb{R}^{d}$ and $\Sigma=\partial\Omega$ . Let $f : \mathbb{R}^{d} \backslash\Omega\to\mathbb{R}$ be a smooth asymptotically flat function such that each connected component of $f ( \Sigma)$ is in a level set of $f$ and $| \nabla f ( x ) | \to\infty$ as $x \to\Sigma$ . Let $( M^{d}, g )$ be the graph of $f$ with the induced metric from $\mathbb{R}^{d} \backslash\Omega\times\mathbb{R}$ and ADM mass $m_{\mathrm{ADM}}$. If each connected component of $\Omega_i$ of $\Omega$ is convex and the $d\nu$-ADEI holds then
\be
m_{\mathrm{ADM}} \geq\sum_{i=1}^{k} \frac{1} {2} \left( \frac{| \Sigma_{i} |} {\omega_{d-1}} \right)^{\frac{d-2} {d-1}} + \frac{K}{(d-1)}. 
\ee
\end{theorem}
\begin{proof}
From Ref.~\cite{lam2010} we have that 
\be
m_{\mathrm{ADM}} \geq\sum_{i=1}^{k} \frac{1} {2} \left( \frac{| \Sigma_{i} |} {\omega_{d-1}} \right)^{\frac{d-2} {d-1}}+\frac{1}{2  (d-1)} \int_{M^{d}} R \, d\nu. 
\ee 
If we assume $d\nu$-ADEI \eqref{eqn:dnadei} where $d\nu$ is given by \eqref{eq:Graphnu} then the result follows. 
\end{proof}

According to \cite[Remark 11]{lam2010}, for spherically symmetric functions $f=f(r)$ on $\mathbb{R}^d$, the mass is always nonnegative and then there are no spherically symmetric asymptotically flat smooth functions on $\mathbb{R}^d$ whose graphs have negative scalar curvature everywhere. Then even if the DEC is violated in this case there is a lower bound for $K$ according to Corollary~\ref{cor:lowerboundK}. 

\subsubsection{Inverse Mean Curvature Flow}
Whenever there is a lack of symmetry, the choice of $d\nu$ becomes more complicated and must be adapted for the specific technique being employed. According to the discussion in \cite[Section 9]{PDE_PI}, in the proof of the Riemannian Penrose inequality using inverse mean curvature flow, one may assume a lower bound for
\be
\int_0^\infty\int_{\Sigma(t)} \sqrt{\frac{|\Sigma(t)|_g}{16\pi}} \frac{R}{16\pi} dA dt \,,
\ee
for a family of surfaces $\Sigma(t)$ evolving under inverse mean curvature flow (see \cite{HuiskenIlmanen2001}). Applying the co-area formula, we have that $dA dt = H d\mu_g$, where $H$ is the mean curvature of $\Sigma(t)$. Since $|\Sigma(t)|_g = |\Sigma(0)|_g e^t$, if the flow exists, one could take 
\be
d\nu = C H e^{t/2} d\mu_g,
\ee
where 
\begin{equation}
C=\frac{1}{8\pi}\sqrt{\frac{|\Sigma(0)|_g}{16\pi}},
\end{equation}
to conclude the following inequality
\be
m_{\mathrm{ADM}} -\sqrt{\frac{\Amin(\Sigma)}{16 \pi}} \geq \frac{1}{2}\int_M R \, d\nu  \,.
\ee
If the inverse mean curvature flow does not exist, one could use the weak solution to Huisken-Ilmanen inverse mean curvature flow $u$ (see \cite{HuiskenIlmanen2001}), which is equals zero on $\Sigma$, and define analogously
\be \label{eq:IMFdnu}
d\nu = C |\nabla u|_g e^{u/2} d\mu_g.
\ee
Therefore, Theorem 7 in~\cite{PDE_PI} can be extended as follows:
\begin{theorem}
    Given an asymptotically flat $(M^3, g)$ with $H_2(M)=0$ and a minimal connected boundary $\Sigma$ which bounds an interior region. If the $d\nu$-ADEI holds then
    \be
    m_{\mathrm{ADM}} \geq \sqrt{\frac{\Amin(\Sigma)}{16 \pi}} + \frac{K}{2} \,.
    \ee
\end{theorem}

\subsubsection{$\mathrm{SU}(d+1)$ invariant initial data sets}

For non-spherically symmetric initial data set and $\mathcal{K}\neq0$, \cite{Khuri:2024rfy} established the Penrose inequality for asymptotically flat $\mathrm{SU}(d+1)$, $d\geq1$, invariant initial data sets and for cohomogeneity one time-symmetric initial data sets. The proof uses the idea of finding $(f, \phi)$ such that $(M, g, \mathcal{K})$ comes from a slice of the static spacetime 
\be
(\mathbb{R} \times M^d, -\phi^2dt^2+\tilde{g})
\ee
and $g_{ij} = \tilde{g}_{ij}- \phi^2 f_i f_j$. To achieve this, the pair $(f, g)$ must satisfy the generalized Jang equation (see ~\cite[Section 6]{PDE_PI}).

For asymptotically flat $\mathrm{SU}(d+1)$, $d \geq1$, invariant initial data sets, with an outermost MOTS boundary $\partial M$, they showed that
\be
m_{\mathrm{ADM}} - \frac{1}{2}\left( \frac{|\partial M|} {\omega_{2 d+1}} \right)^{\frac{2 d} {2 d+1}}  \geq \frac{1}{(2d+1) \omega_{2d+1} }\int_{M} (\mu - |J|_g) \phi dV_{\tilde{g}}.
\ee
Assuming the pointwise DEC, this leads to the classical Penrose inequality. Analogously, one could assume a lower bound for this quantity to derive a modified Penrose inequality. Therefore, we could set
\be\label{eq:Cohomdnu}
d\nu = \frac{\phi}{\omega_{2d+1}} \, dV_{\tilde{g}},
\ee
and obtain the following result.

\begin{theorem}
Let $( M^{2 ( d+1 )}, g, \mathcal{K})$ , $d \geq1$ be an asymptotically flat $\mathrm{SU}(d+1)$ invariant initial data set, with outermost MOTS boundary. If the $d\nu$-ADEI, as given in \eqref{eq:Cohomdnu}, holds, then 
\be
m_{\mathrm{ADM}} \geq \frac{1}{2} \left( \frac{|\partial M|} {\omega_{2 d+1}} \right)^{\frac{2 d} {2 d+1}} + \frac{K}{(2d+1)} \,.
\ee
\end{theorem}
\begin{remark}
In the time-symmetric case, \cite{Khuri:2024rfy} extend the result to cohomogeneity-one initial data sets satisfying $R\geq 0$. However, this result does not generalize easily to our setting, as its proof relies on the positive mass theorem, which in turn requires pointwise nonnegative scalar curvature.  
\end{remark}

What is clear from these examples is that whenever there is a lack of symmetry but a natural choice of $d\nu$ exists, the measure $d\nu$ is adapted to the type of proof being applied. We should note that these conditions still allow for violations of the pointwise DEC.

\section{Dynamical black holes and weak cosmic censorship}
\label{sec:dynamical}

The preceding part of this work focused on proving a Penrose inequality of the form of Eq.~\eqref{eqn:peninitial} for initial data sets. From earlier results, it is known that averaged energy conditions are sufficient to prove null geodesic incompleteness \cite{Fewster:2010gm, Fewster:2019bjg}, and in \cite{Kontou:2023ntd} they are used to bound the evaporation rate of black holes. This can point to a generalized form of Penrose's original motivation; the connection with weak cosmic censorship.

\subsection{Non-evaporating dynamical black holes}

The classical Penrose inequality was first proposed for spacetimes satisfying the NCC \eqref{eqn:NCC} or the NEC along with Einstein's equations. This condition plays a crucial role in allowing one to invoke Hawking's black hole area theorem to establish the inequality. 

The S-DEC is motivated by potential violations of the DEC, particularly near the horizon. It is also clear that the S-DEC permits violations of the NEC. Given this context, one might conclude that the S-DEC could be either incompatible with or too weak to support Penrose's original heuristic argument, as it does not allow for the application of Hawking's black hole area theorem.

A potential way to circumvent this issue is to apply the recently established version of Hawking's black hole area theorem by \cite{Kontou:2023ntd}, which permits violations of the NCC. 

To introduce the relevant averaged energy condition, we first recall the criterion for the formation of a focal point, a point where a Jacobi field becomes zero on null geodesics emanating normally from a spacelike submanifold $P$. The test to determine the formation or not of a focal point for causal geodesics is the sign of the Hessian of the action integral $E$ of the geodesic. 

In particular for a curve \(\gamma : [0, \ell] \rightarrow \overline{M}\) affinely parameterized by \(\lambda\), \(E\) is defined as
\be
E[\gamma] := \frac{1}{2}\int_{0}^{\ell} \overline{g}(\gamma'(\lambda), \gamma'(\lambda))d\lambda \,.
\ee
Now consider the set of all piecewise smooth curves joining \(P\) to a point \(q\),  leading to a family of curves \(\gamma_s(\lambda):=\zeta(\lambda, s)\) which varies smoothly in $s$. The tangent and the transverse vector fields are defined as $U_\mu=\gamma'(\lambda)$ and $V_\mu=\partial \gamma_s/\partial s |_{s=0}$. The first variation of  $E[\gamma_s]$ is zero for geodesics while its second variation is
\be
	\label{eq:hessian}
	\mathcal{H}[V]\equiv \frac{\partial^2E[\gamma_s]}{\partial s^2}\Big\vert_{s = 0} = 
	\int_{0}^{\ell} \left[(\nabla_UV_{\mu})(\nabla_UV^{\mu}) + R_{\mu\nu\alpha\beta}U^{\mu}V^{\nu}V^{\alpha}U^{\beta}\right] d\lambda -U^{\mu} \nabla_V V_{\mu} \Big\vert_0^{\ell}\,.
\ee
What is left is to choose a convenient $V$, Let \(e_i\) with \(i = 1, \ldots, n - 2\) be an orthonormal basis of \(T_{\gamma(0)}P\), and parallel transport it along \(\gamma\) to generate \(\{E_i\}_{i = 1, \ldots, n-2}\). Then, take \(f\) a smooth function with \(f(0) = 1\) and \(f(\ell) = 0\). Calculating the sum of Hessians for all $fE_i$ gives
\begin{equation}
	\label{eq:hessian-averagded}
	\sum_{i=1}^{n - 2}\mathcal{H}\left[fE_i\right] = - \int_{0}^{\ell} \left((n - 2)f'(\lambda)^2 - f(\lambda)^2R_{\mu\nu}U^{\mu}U^{\nu}\right) d\lambda - (n - 2)U_{\mu}H^{\mu}\Big\vert_{\gamma(0)},
\end{equation}
where $H^{\mu}$ is the mean curvature vector field of $P$. The expression of the right-hand side of 
Eq.~\eqref{eq:hessian-averagded} may be written in invariant form by regarding $f$ as the coordinate expression of a density $f$ of weight $-1/2$ on $\gamma$.

Now we can apply the condition for the formation of a focal point. If
\be
\mathcal{H}[V] \geq 0 \,,
\ee
for any $V$ then there is a focal point to $P$ along $\gamma$. If the inequality holds strictly then the focal point is located before $q$ \cite{o1983semi}. Then we have \cite{Fewster:2019bjg, Kontou:2023ntd}
\begin{prop}
		\label{prop:fp-criteria}
		Let \(P\) be a spacelike submanifold of \(\overline{M}\) of co-dimension \(2\) with mean normal curvature $H^\mu$ and let \(\gamma\) be a null geodesic with tangent vector $U^\mu$ joining \(p \in P\) to \(q\in J^+(P)\). Let $\gamma$ be affinely parametrized by $\lambda \in [0,\ell]$.  If there exists a smooth \((-\frac{1}{2})\)-density \(f\) which is non vanishing at \(p\) but is null at \(q\), and such that
		\begin{equation}
		\label{eq:fp-criteria}
		\int_0^\ell \big((n -2)f'(\lambda)^2 - f(\lambda)^2 R_{\mu \nu} U^\mu U^\nu \big)d\lambda \le -(n -2) f(0)^2 U_\mu H^\mu \big|_{\gamma(0)}
		\end{equation}
		then there is a focal point to \(P\) along \(\gamma\). If the inequality holds strictly then the focal point is located before \(q\).
\end{prop}
If $U_\mu H^\mu<0$ and\footnote{Assuming that the outward null direction is given by $U$, a standard computation shows that $U_\mu H^\mu = f \theta_+$ on $P$ for some smooth function $f>0$ (see Eq.~\eqref{eq:MOTS_MITS}).} the NCC holds, a focal point is formed on the null generators of the event horizon \cite{Hawking:1971vc}. As this is forbidden by the causality conditions, we have $U_\mu H^\mu \geq0$ and so the event horizon area always increases. In Ref.~\cite{Kontou:2023ntd} it was noted that the same conclusion can be drawn if
\be
\label{eqn:Jint}
\inf_{f\in C^\infty_{1,0}[0,\ell]} J_\ell[f] \leq 0 \,,
\ee
with 
\be
J_\ell[f]=\int_0^\ell\left((n-2)f'(\lambda)^2 -f(\lambda)^2 R_{\mu \nu}U^\mu U^\nu \right) d\lambda \,,
\ee
and the infimum in \Cref{eqn:Jint} is taken over all smooth $(-\frac{1}{2})$-densities on the interval $[0, \ell]$, $\ell>0$ such that $f(0)=1$ and $f(\ell)=0$. Of course this condition is hiding an energy condition as you need some information for the null energy to deduce the sign of the integral. For now we will keep it general.

We note that a condition of the form of Eq.~\eqref{eqn:Jint} can be compatible with S-DEC. Now assume that our dynamical spacetime settles to a static one\footnote{It could be for example a black hole merger.}. Recall that $(\overline{M}, \overline{g})$ is a static, spherically symmetric spacetime if the metric can be expressed as
\be
\overline{g} = -h(r) dt^2 + f(r) dr^2 + r^2 d\Omega^2,
\ee
where $h(r)$ and $f(r)$ are smooth functions. Additionally, assuming that $(\overline{M}, \overline{g})$ is a black hole spacetime with an event horizon $\mathcal{H}$, the S-DEC can be extended to this class of spacetimes. Specifically, if there exist asymptotically flat spherically symmetric initial data set $(M, g, \mathcal{K})$, $g$ as in \eqref{eq:SphrSymMetric}, such that the initial data with boundary given by $\mathcal{H} \cap M$ and satisfies the S-DEC with respect to the outermost MITS (or $\mathcal{H}\cap M$ in its absence), we say that $(\overline{M}, \overline{g})$ satisfy the S-DEC\footnote{In order to simplify the arguments, we assume that spherically symmetric metrics are given by \eqref{eq:SphrSymMetric} instead of \eqref{eq:SphrMetricds}.}. Now we can state the following theorem

\begin{theorem}
\label{the:classarea}
Suppose that
\begin{enumerate}
    \item $( \overline{M}^4, \overline{g})$ is a strongly asymptotically predictable spacetime.
    \item Let $M_0$ be an asymptotically flat partial Cauchy surface with ADM mass $m_{\mathrm{ADM}}$ for the globally hyperbolic region and containing a MOTS $\Sigma$.
    \item For $\mathcal{H}$ the event horizon, $U$ the tangent field of its null generator $\gamma(\lambda)$, the inequality of Eq.~\eqref{eqn:Jint} holds.\label{item:EnergyCond}
    \item The Bondi mass $m_B$ does not increase to the future and the Bondi mass approaches\footnote{Since gravitational waves carry positive energy~\cite{Bondi1962-rc, Sachs1962-hs}, the Bondi mass cannot increase towards the future. Its convergence to the ADM mass is only known under additional assumptions, see~\cite{Hayward2003, Kroon2003-oh, Huang2007-iy}.} $m_{\mathrm{ADM}}$.\label{item:BondiMass}
\item
One of the following two conditions holds
\begin{enumerate}[label=(\roman*),ref=\theenumi.(\roman*)]
\item\label{item:SettlesStaticSDEC}
 the spacetime settles to a spherically symmetric static black hole satisfying the S-DEC.
\item\label{item:SettlesKerr} the spacetime settles to a Kerr metric.
\end{enumerate}
\end{enumerate}
then 
\be
\label{eqn:pendyn}
2 m_{\mathrm{ADM}} \geq \sqrt{\frac{\Amin(\Sigma)}{4 \pi}},
\ee
where $\Amin(\Sigma)$ is the minimum area required to enclose the $\Sigma$ 
\end{theorem}
\begin{proof}
Analogously to the NEC, the energy condition~\ref{item:EnergyCond} implies that MOTSs are contained in the black hole region, that is, $\Sigma$ lies behind the event horizon $\mathcal{H}$~\cite[Proposition A.2]{Kontou:2023ntd}. Since $\Amin(\Sigma)$ is the minimum area required to enclose $\Sigma$, one has $\Amin(\Sigma)\leq |\mathcal{H} \cap M_0|$.

The energy condition~\ref{item:EnergyCond} also ensures that the Hawking black hole area theorem from~\cite[Theorem III.2]{Kontou:2023ntd} applies. Thus, the area of the event horizon is non-decreasing, implying that $|\mathcal{H}\cap M_0| \leq |\mathcal{H}_\infty|$, where $\mathcal{H}_\infty$ denotes the event horizon at the future null infinity. 

Assuming \ref{item:SettlesStaticSDEC}, in the static spacetime any cross section of the horizon is non-expanding. Consequently, in the corresponding initial data set, it is a MOTS, and condition \ref{item:EnergyCond} guarantees that it is an outermost MOTS. The S-DEC for the spacetime ensures that the modified Penrose inequality also holds for the initial data set associated with the spacetime and then it also holds for $\mathcal{H}_\infty$. Even if there is an outermost MITS for this initial data set, we can apply Eq.~\eqref{eq:MITSargument} omitting the $A_{\min}$ due to spherical symmetry, following the same reasoning as in the paragraph before Eq.~\eqref{eq:AminArg}. Moreover, the mass $m_{\mathrm{ADM}}$ obtained with the Penrose inequality of the final black hole coincides with the Bondi mass $m_B$ of the original spacetime. Then
\be
2m_{\mathrm{ADM}}  \geq 2m_{B}  \geq \sqrt{|\mathcal{H}_\infty|/4\pi} \geq \sqrt{|\mathcal{H}\cap M_0|/4\pi} \geq \sqrt{\Amin(\Sigma)/4\pi}\,,
\ee
where we have used the fact that the Bondi mass is less than or equal to the ADM mass, given that condition~\ref{item:BondiMass} holds.
If ~\ref{item:SettlesKerr} holds instead, the first two paragraphs, combined with the original Penrose argument and the inequalities in Eq.~\eqref{eqn:Penarg} yield the result. In both cases the inequality of Eq.~\eqref{eqn:pendyn} holds. 
\end{proof}

\subsection{Evaporating Black Holes}

 While the fully semiclassical description of an $n$-dimensional evaporating black hole has not been developed yet, it is reasonable to assume that the ADM mass of an evaporating black hole spacetime remains constant and thus we can derive a meaningful version of the Penrose inequality. That is because the ADM mass arises as a conserved quantity via the Hamiltonian formulation \cite{Bilal:1993wm}. To ensure that, some approaches \cite{Massar:1999wg} set it as a constant of motion for the dynamical evolution. We emphasize that we remain agnostic about the final state of black hole evaporation and instead examine a spacetime where a horizon is always present but is allowed to shrink. Such a spacetime would describe all stages of evaporation except the final and remain globally hyperbolic.

To allow evaporation we can weaken the condition of Eq.~\eqref{eqn:Jint}. In particular, we do not assume that the infimum of the integral is non-positive. Then the conclusion of the theorem is not that the area of the black hole horizon is non-decreasing but that the rate of change of the area is bounded \cite[Theorem IV.1]{Kontou:2023ntd}
\be \label{eq:evapbound_causality}
\left.\frac{d}{d\lambda}\right|_{\lambda=0}|\mathcal{H}_\lambda|=\int_{\mathcal{H}_0} H^\mu U_\mu \geq-\frac{1}{n-2}\left(\inf _{f \in C_{1,0}^{\infty}[0, \ell]} J_{\ell}[f]\right) \cdot |\mathcal{H}_0|,
\ee
where $\mathcal{H}_\lambda$ is a cross section defined by the null generators at $\lambda$.

It is widely known that whenever a black hole spacetime satisfies the NEC, then any trapped surface $T$ must lie inside the black hole region $B$ \cite{Hawking:1971vc}. Nevertheless, for evaporating black holes this statement is not true anymore, trapped surfaces are allowed to sit outside the event horizon (e.g. the apparent horizon of a Vaidya spacetime that violates the NEC lies outside the event horizon). However, one can prove singularity theorems if the initial condition is strengthened \cite{Fewster:2019bjg}. Instead of requiring a marginally trapped surface a \textit{sufficiently trapped} surface is needed.

As a consequence, the sufficiently trapped surface lies inside the event horizon \cite[Proposition A.2]{Kontou:2023ntd}. Precisely, focal points form along every normal null geodesic $\gamma$, parametrized by $\lambda$, emanating normally from $T$ if, and only if,
\be
\label{eqn:suftrapped}
U_{\mu} H^{\mu} \big|_{\gamma( 0 )} \leq-\frac{1} {n-2} \int_{\gamma} \big( ( n-2 ) f^{\prime} ( \lambda)^{2}-f ( \lambda)^{2} R_{\mu\nu} U^{\mu} U^{\nu} \big) d \lambda\,, 
\ee
holds for any $U^\mu$. This is necessary to prove null geodesic incompleteness. Whenever this condition is satisfied, we say that a trapped surface $T$ is \textit{sufficiently trapped}.

For dynamical horizons, we expect that, in the future, Hawking radiation will begin to dominate over the infall of matter. As a result of this radiation, the area of the event horizon is expected to decrease. Geometrically, this can be characterized by the fact that the cross sections of the event horizon are weakly outer trapped surfaces, i.e., $\theta_+\leq 0$. Inspired by this idea, an initial data set $(M, g, \mathcal{K})$ in a black hole spacetime $(\overline{M}, \overline{g})$ is called an \textit{evaporating initial data set} if its intersection with the event horizon, $\mathcal{H}\cap M$, is a weakly outer trapped surface ($\theta_+\leq0$). Recall that spacetimes satisfying the NEC prohibit $\mathcal{H}\cap M$ to be outer trapped. Thus, evaporating initial data sets are mostly relevant in spacetimes where the event horizon can shrink, potentially due to evaporation effects.

For $d$-dimensional asymptotically flat evaporating initial data sets with $d<8$, it is always possible to find an outermost MOTS, $\Sigma$ enclosing the weakly outer trapped surface $\mathcal{H}\cap M$ \cite{Andersson:2007gy,Eichmair2009, Andersson2011}. For asymptotically flat evaporating initial data sets with spherical symmetry, $\Sigma$ is necessarily a sphere and the S-DEC is defined with respect to the outermost MITS enclosing $\Sigma$ (or $\Sigma$ in absence of it).

\begin{theorem}
\label{the:PIevapo}
Suppose that
\begin{enumerate}
    \item $( \overline{M}^{n}, \overline{g})$, $n<9$, is a strongly asymptotically predictable spacetime with a sufficiently trapped surface $T$.
    \item Let $M_0$ and $M_1$ be $d$-dimensional asymptotically flat partial Cauchy surfaces for the globally hyperbolic region such that $T \subset M_0$ and $M_0 \subset I^-(M_1)$. Further assume that $M_1$ is an evaporating spherically symmetric initial data set with boundary $\mathcal{H}\cap M_1$ and metric as Eq.\eqref{eq:SphrSymMetric}.
\end{enumerate}
If $M_1$ satisfy the S-DEC and there exists $s>0$ such that the cross section $\mathcal{H}_s$ coincides with $\mathcal{H} \cap M_1$, then
\be
\label{eqn:PIevap}
m_{\mathrm{ADM}} \geq \frac{1}{2}\left({\frac{\Amin(T)}{ \omega_{n-2} }}\right)^\frac{n-3}{n-2}  \exp \left(-\frac{1}{2(n-2)}\int_{0}^s \nu (\lambda)d\lambda \right).
\ee
where $\Amin(T)$ is the minimum area required to enclose $T$ and $\nu (\lambda):=\inf_{f \in C_{1, 0}^{\infty} [\lambda, \ell]} J_{\ell} [ f ]$
\end{theorem}
\begin{proof} 
By hypothesis, $T$ satisfies the condition required to establish null geodesic incompleteness, i.e., $T$ is sufficiently trapped \eqref{eqn:suftrapped}. Consequently, $T$ lies behind the event horizon $\mathcal{H}$~\cite{Kontou:2023ntd}, implying that $\Amin(T) \leq |\mathcal{H} \cap M_0|$.

We can extend \eqref{eq:evapbound_causality} to all cross sections ${\mathcal{H}_\lambda}$ by
\be
\frac{d}{d\lambda} |\mathcal{H}_\lambda|=\int_{\mathcal{H}_\lambda} H_\lambda^{\mu} U_{\mu} \geq-\frac{1} {n-2} \nu (\lambda) |\mathcal{H}_\lambda|,
\ee
where $\nu(\lambda)$ is given by
\be
\nu(\lambda):=\inf_{f \in C_{1, 0}^{\infty} [\lambda, \ell]} J_{\ell} [ f ].
\ee
The above inequality can be written as 
\be
\frac{d}{d\lambda} (\ln( |{\mathcal{H}_\lambda}|) \geq-\frac{1} {n-2} \nu (\lambda),
\ee
integrating both sides from any $s>0$ to $0$ and exponentiating, we obtain that the area evolves along the null generators according to the following inequality
\be
|\mathcal{H}_s|  \geq |\mathcal{H}_{0}| \exp\left(-\frac{1}{n-2}\int_{0}^s \nu (\lambda)d\lambda \right).
\ee
Combining that $T$ lies behind the horizon $\mathcal{H}$ together with $\mathcal{H}\cap M_0 \equiv \mathcal{H}_0$, we obtain
\begin{align}
\Amin(T)\leq |\mathcal{H} \cap M_0| & \leq |\mathcal{H}_s| \exp\left(\frac{1}{n-2}\int_{0}^s \nu (\lambda)d\lambda \right)
\end{align}
for all $s>0$.

Since there exists $s>0$ such that $\mathcal{H}_s \equiv \mathcal{H} \cap M_1$, and $\mathcal{H}\cap M_1$ is weakly outer trapped due to $M_1$ being an evaporating initial data set, we can find an outermost MOTS $\Sigma\subset M_1$ enclosing $\mathcal{H} \cap M_1$ \cite{Andersson:2007gy,Eichmair2009, Andersson2011}. Momentarily, suppose that there is no outermost MITS outside $\Sigma$. Applying the modified Penrose inequality, we obtain the follow chain of inequalities,
\bea
\Amin(T)\leq |\mathcal{H} \cap M_0| & \leq & |\mathcal{H} \cap M_1| \exp \left(\frac{1}{n-2}\int_{0}^s \nu (\lambda)d\lambda\right) \leq |\Sigma|\exp\left(\frac{1}{n-2}\int_{0}^s \nu (\lambda)d\lambda\right) \nonumber \\
&\leq & \omega_{n-2} (2 m_{\mathrm{ADM}})^{\frac{n-2}{n-3}} \exp\left(\frac{1}{n-2}\int_{0}^s \nu (\lambda)d\lambda\right),
\eea
where $|\mathcal{H} \cap M_1|\leq |\Sigma|$ holds since both are spheres and the metric is given by $\eqref{eq:SphrSymMetric}$. Rearranging the terms, we have
\be
m_{\mathrm{ADM}} \geq \frac{1}{2}\left({\frac{\Amin(T)}{ \omega_{n-2} }}\right)^\frac{n-3}{n-2}  \exp \left(-\frac{1}{2(n-2)}\int_{0}^s \nu (\lambda)d\lambda \right)\,,
\ee
which completes the proof for this case. Whenever there is an outermost MITS outside $\Sigma$, we can apply the same argument to the outermost MITS and use the inequality \eqref{eq:MITSargument} without $A_{\min}$, as justified in \Cref{the:classarea}.
\end{proof}

We should briefly comment on the meaning of this result and its differences from the original Penrose derivation. First, the surface considered is not a MOTS but a sufficiently trapped surface. Second, we compare two initial data sets at different times, not the end of the evaporation process (as this remains outside the realm of semiclassical gravity). Third, the inequality now includes information about the horizon evolution through the $\nu(\lambda)$. Nevertheless, the result shows that if the sufficiently trapped surface, required for geodesic incompleteness, lies inside the event horizon then a modified Penrose inequality must hold between initial data sets as a test of weak cosmic censorship.

In the next section, we will use this result to bound the area of sufficiently trapped surfaces in terms of the ADM mass, assuming an energy condition that could be satisfied by quantum fields.

\section{Applications}
\label{sec:applications}

\subsection{A condition that prevents evaporation}

We can replace the NCC with an average energy condition of the form 
\be
\label{eqn:classcond}
\int_{\gamma} f ( \lambda)^{2} R_{\mu \nu} U^\mu U^\nu d \lambda\geq( n-2 ) \| f^{\prime} \|^{2} \,.
\ee
According to \cite{Kontou:2023ntd}, this is the weakest condition that does not allow the horizon area to decrease and so forbids black hole evaporation\footnote{At least when one considers the proof using the formation of focal points}.  

Specifically, in Ref.~\cite{Kontou:2023ntd} it was shown that a condition of the form of \eqref{eqn:classcond} is sufficient to prove Eq.~\eqref{eqn:Jint}. Then we have the following corollary: 

\begin{corollary}
    If in Theorem~\ref{the:classarea}, in condition $3$, we replace Eq.~\eqref{eqn:Jint} with Eq.~\eqref{eqn:classcond}, the theorem holds.  
\end{corollary}

This theorem shows that the original Penrose argument holds for weaker conditions than the NCC. However the condition of Eq.~\eqref{eqn:classcond} is not generally obeyed by quantum fields. The bound is always nonnegative thus incompatible with quantum field theory \cite{Epstein1965}.

\subsection{A condition inspired by QEIs}

The authors of \cite{Kontou:2023ntd} proposed conditions inspired by Quantum Energy Inequalities (QEIs) that involve compactly supported densities and the Ricci curvature. We will briefly introduce these conditions to obtain bounds on the evaporation rate of the black hole and apply them to our Penrose-type inequality bound, as well as revisit the examples presented in their work. For a detailed explanation of the conditions and examples, we refer to \cite{Fewster:2019bjg,Kontou:2023ntd}.

Let $P\subset \overline{M}$ be a spacelike submanifold of codimension 2 with mean curvature vector field $H_{\mu}$. For every point $p$ such that $P$ is future-converging, i.e., $H^\mu(p)$ is past-pointing timelike everywhere on $P$, we may write $H^\mu = H \hat{H}^\mu$, where $H<0$ and $\hat{H}^\mu$ is future-pointing timelike unit vector. For any future-direct null geodesic $\gamma$ emanating normally from $P$, we can fix\footnote{Fixing an affine parametrization simplifies the presentation of the energy condition, the invariant form and further details can be found in \cite[Section 4.2]{Fewster:2019bjg}.} an affine coordinate $\lambda$ on $\gamma$, such that $p = \gamma(0)$ and $\hat{H}_{\mu} d\gamma^{\mu}/d\lambda=1$. Then for every smooth compactly supported $(-1/2)$-density $g$ on $\gamma$, and every choice of smooth nonnegative constants $Q_0(\gamma)$ and $Q_1(\gamma)$ we assume
\be
\label{eqn:QEI}
\int_{-\infty}^{\infty} g(\lambda)^2 R_{\mu\nu}U^{\mu}U^{\nu}\, d\lambda \ge -Q_1(\gamma) \vert\vert g^\prime \vert\vert^2 - Q_0(\gamma) \vert\vert g\vert\vert^2 \,.
\ee
We would like to apply this inequality to compute $\nu(\lambda)$ in Eq.~\eqref{eq:evapbound_causality}. However, the condition for the formation of focal points and the above inequality are not directly compatible. The smooth function $f \in C^\infty_{1,0}$ in \eqref{eq:evapbound_causality} does not satisfy the same boundary conditions as $g$ in \eqref{eqn:QEI}. This issue can be resolved by setting $g = f \varphi$, where $\varphi$ is a smooth compactly supported function such that $\varphi(\ell_0)=1$ for some $0<\ell_0<\infty$. However, introducing this additional function imposes an extra requirement: we must estimate the null energy for $0 \leq \lambda \leq \ell_0$. To simplify the analysis, we assume a pointwise condition
\be
\label{eqn:pointwise}
R_{\mu \nu} U^\mu U^\nu\big|_{\gamma(\lambda)} \geq \rho_0\,, \quad\quad \forall\lambda\in [0, \ell_0] \,,
\ee
where $\rho_0 \in \mathbb{R}$. Since $\rho_0$ is allowed to be negative for finite values of the affine parameter on the geodesic, this condition is still less restrictive than the NEC.

Now we have to pick the functions $f$ and $\varphi$. A computation, following \cite[Section 4.2]{Kontou:2023ntd}, shows that\footnote{In Ref.~\cite{Kontou:2023ntd} there is also a computation for general $n$ and general number of derivatives.}
\begin{equation}
J_{\ell}[f] \leq-\rho_0 \ell_0+\rho_0\|\varphi\|^2+(n-2)\left\|f^{\prime}\right\|^2+Q_1(\gamma)\left\|g^{\prime}\right\|^2+Q_0(\gamma)\|g\|^2,
\end{equation}
where $g = f\varphi$. Moreover, by solving an optimization problem, they showed that it is possible to find $f$ and $\varphi$ such that $J_\ell[f] \leq \nu(Q_0(\gamma),Q_1(\gamma), \ell, \ell_0, \rho_0)$ (see \cite[Theorem 4.3]{Kontou:2020bta}). 

Taking $\ell\rightarrow \infty$, $\rho_0<0$, and optimizing for $\ell_0$, they obtained a bound for four-dimensional ($n=4$) black hole spacetimes independent of $\ell$ and $\ell_0$. Additionally, whenever $Q_0(\gamma)\leq Q_0$ and $Q_1(\gamma)\leq Q_1$ for constants $Q_0$ and $Q_1$ independent of $\gamma$, and the lower bound \eqref{eqn:pointwise} holds along every null geodesic, $\nu$ can be chosen to be independent of $\gamma$. Hence, we have
\be
\left.\frac{d}{d\lambda}\right|_{\lambda=0}|\mathcal{H}_\lambda|=\int_{\mathcal{H}} H^{\mu} U_{\mu} \ge - \frac{1}{2}\nu_{\text{opt}}(Q_1,Q_0,\rho_0)\cdot |\mathcal{H}_0| \,,
\ee
where
\bea
\label{eqn:nuopt}
\nu_{\text{opt}}(Q_0,Q_1,\rho_0)&=&\sqrt{Q_0(2+Q_1)}+\sqrt{Q_1 Q_0} \nonumber \\
&&+|\rho_0| \sqrt{\frac{Q_1}{Q_0+\rho_0}}\arcsinh{\left(\sqrt{\frac{Q_0 + \rho_0}{|\rho_0|}}\right)} \,.
\eea
The constants $Q_0(\gamma)$ and $Q_1(\gamma)$ encode the specifics of the matter model through the relevant energy condition.  The pointwise bound $\rho_0$ is generally undetermined as quantum fields have no pointwise lower bounds for all states, but estimates can be obtained close to the horizon, see next section. An immediate corollary of \Cref{the:PIevapo} is the following.

 \begin{corollary}
 \label{cor:PIevapoqei}
Suppose that
\begin{enumerate}
    \item $( \overline{M}^{4}, \overline{g})$ is a strongly asymptotically predictable spacetime with a sufficiently trapped surface $T$.
    \item Let $M_0$ and $M_1$ be $d$-dimensional asymptotically flat partial Cauchy surfaces for the globally hyperbolic region such that $T \subset M_0$ and $M_0 \subset I^-(M_1)$. Further assume that $M_1$ is an evaporating spherically symmetric initial data set with boundary $\mathcal{H}\cap M_1$ and metric as Eq.\eqref{eq:SphrSymMetric}.
    \item For $\mathcal{H}$ the event horizon, $U$ the tangent field of its null generator $\gamma(\lambda)$, the following inequality    \be 
    \int_0^{\infty} g (\lambda)^2 R_{\mu\nu}U^{\mu}U^{\nu} \ge -Q_1 \vert\vert g^\prime \vert\vert^2 - Q_0 \vert\vert g\vert\vert^2 \,
    \ee
    holds with $Q_1(\gamma)\leq Q_1$ and $Q_0(\gamma)\leq Q_0$. Additionally, there exists a finite constant $\rho_0<0$ such that along every null generator
    \be
    R_{\mu \nu} U^\mu U^\nu\big|_{\gamma(\lambda)} \geq \rho_0\,, \quad\quad \forall\lambda\in [0, \ell_0] \,,
    \ee
    for some $\ell_0>0$.
\end{enumerate}
If $M_1$ satisfy the S-DEC and there exists $s>0$ such that the cross section $\mathcal{H}_s$ coincides with $\mathcal{H} \cap M_1$, then 
\be
m_{\mathrm{ADM}} \geq\sqrt{\frac{\Amin(T)}{16\pi}}  \exp\left(-\frac{1}{4}\nu_{\text{opt}}(Q_1,Q_0,\rho_0)\cdot s \right) .
\ee
where $\Amin(T)$ is the minimum area required to enclose $T$ and 
 $\nu_{\text{opt}}$ is given by \eqref{eqn:nuopt}.
\end{corollary}

\subsection{Specific models}

To apply the above result, one needs estimates for $\rho_0$, $Q_0$ and $Q_1$. Estimates of the null energy density around a black hole horizon may suggest a reasonable value for $\rho_0$. For this, we use the numerical results of Levi and Ori \cite{LeviOri2016} for the renormalized stress-energy tensor of a minimally coupled scalar field near the horizon of a Schwarzschild black hole. They find the value of the null energy on the $r=3M$ geodesics given by
\be
\label{eqn:rhozerofin}
\rho_0\approx- 0.1 \frac{k^4}{\hbar^3 c^3} T^4 \,.
\ee
in SI units, where $k$ is the Boltzmann constant.

\subsubsection{The non-minimally coupled classical Einstein–Klein–Gordon theory}
A classical example that violates the NEC is the non-minimally coupled scalar field \cite{Kontou:2020bta}. These scalar fields are described by the Lagrangian density
\be
\mathcal{L}[\phi] = -\frac{1}{2}\left[\nabla_{\mu}\phi\nabla^{\mu}\phi   - (m^2 - \xi R)\phi^2\right] \,,
\ee
where $\xi$ is the coupling constant and $m$ the mass of the field. 

Given certain constraints and the numerical value of $\rho_0$, it is possible to obtain a reasonable estimate for $\nu_{\text{opt}}$ \cite{Kontou:2023ntd} based on a method developed in \cite{Brown:2018hym}. Assume that the non-minimally coupled scalar field obeys the bound $\phi^2 < 1/8\pi\xi$ and that the coupling constant is $\xi \in [0,1/4]$, where $1/4$ is always smaller that the conformal coupling $\xi_c= 1/6$. Under these conditions, and assuming the Einstein's equations are satisfied, the field satisfies the energy condition with $Q_0=Q$ and $Q_1=Q \tilde{Q}^2$ independent of $\gamma$
\be
        \int_{\gamma}g^2 R_{\mu\nu}U^{\mu}U^{\nu} \ge -Q\left(\|g'\|^2 + \tilde{Q}^2 \vert\vert g\vert\vert^2\right) \,,
\ee
    where 
\be
    Q = \frac{32\pi\xi\phi_{\max}^2}{1 - 8\pi\xi\phi_{\max}
    ^2}\,,
    \quad\quad
    \tilde{Q} = \frac{8\pi\xi\phi_{\max}\phi'_{\max}}{1 - 8\pi\xi\phi_{\max}^2} \,.
\ee
and
\be
\phi_{\max}=\sup_{\gamma}|\phi|\,, \quad \phi'_{\max}=\sup_{\gamma}\vert \phi'(\lambda) \vert \,.
\ee
In order to estimate the values of $\phi_{\max}$ and $\phi'_{\max}$ we connect the scale of the field magnitude with a temperature by following the hybrid approach of Refs.~\cite{Brown:2018hym, Fewster:2019bjg, Kontou:2023ntd}. We take for $\phi_{\max}$ as the value of the Wick square, the renormalized expectation value of the two-point function $\langle \nord{\phi^2} \rangle_\omega$ in Minkowski spacetime, where \(\omega\) is a thermal equilibrium KMS state \cite{haag2012local}, and connect $\phi_{\max}$ to a temperature $T$. Then for $n=4$ we have
\be
\phi_{\max}^2\sim \frac{T^2}{12} \,, \qquad (\phi'_{\max})^2\sim \frac{2\pi^2 T^4}{45} \,.
\ee
Together with the numerical value of $\rho_0$, it was found that
\be
\nu_{\text{opt}}(T,\xi)=\sqrt{\frac{\pi^3}{15}} \sqrt{\xi} \left(\frac{k}{\hbar T_{\text{pl}}^2} \right)T^3+\mathcal{O}(T/T_{\text{pl}})^4 \,.
\ee
is the Planck temperature, $T$ a temperature related with $\phi$ and $k$ is the Boltzmann constant. Then, for non-minimally coupled scalar fields, we have the following Penrose inequality\footnote{Here we have restored the units inside the exponential.},
\be
m_{\mathrm{ADM}} \gtrsim \sqrt{\frac{\Amin(T)}{16\pi}}  \exp \left(-\frac{s}{4}\sqrt{\frac{\pi^3}{15}} \sqrt{\xi} \left(\frac{k}{\hbar T_{\text{pl}}^2} \right)T^3 \right)   .
\ee
It is important to emphasize that this largely relies on the assumptions made regarding the pointwise energy condition and $\rho_0$. Furthermore, the energy condition pertains to a classical field, only providing an analogy for how this condition might manifest in a quantum non-minimally coupled scalar field. Consequently, this result primarily demonstrates how our theorem can provide a version of the Penrose inequality in cases where the NEC is violated.

\subsubsection{The smeared null energy condition}

In this example we will use the smeared null energy condition (SNEC) \cite{Freivogel:2018gxj,Freivogel:2020hiz} as our energy assumption. As it was shown in Ref.~\cite{Fewster:2002ne} there are no lower bounds of the null energy in quantum field theory when the integral is over finite segments of a null geodesic. This issue was circumvented in \cite{Freivogel:2018gxj} where a $UV$ cutoff was introduced. Then the lower bound of the inequality depends on this cutoff. The SNEC, obeying the semiclassical Einstein equation \eqref{eqn:SEE}, can be written as
\be
\label{eqn:snec}
        \int_{\gamma}g^2(\lambda) R_{\mu\nu}U^{\mu}U^{\nu} \ge -32\pi B \|g'\|^2  \,,
\ee
which is exactly of the form of Eq.~\eqref{eqn:QEI}.

The SNEC has been proven for minimally coupled quantum scalar fields only for Minkowski spacetime \cite{Fliss:2021gdz}. The value of $B$ depends on the cutoff scale $\luv$ as
\be
\label{eqn:GNUV}
G_N \lesssim \luv^{n-2} \,.
\ee 
When this inequality is saturated, the $\luv$ is the Planck length scale, and $B=1/32\pi$, which is its maximum value. For $\luv \gg \ell_{\text{pl}}$, for example in the case of an effective field theory, $B \ll 1$. Thus $B$ can be thought as a dimensionless proportionality constant of Eq.~\eqref{eqn:GNUV}. For more details see Ref.~\cite{Freivogel:2020hiz}.

Using this inequality as well as the estimated value of $\rho_0$, we have for $\nu_{\text{opt}}$ of Eq.~\eqref{eqn:nuopt}
\be
\nu_{\text{opt}}(B,T)=\frac{2\sqrt{2\pi^3 }}{\sqrt{5}} \sqrt{B} \left(\frac{k}{\hbar T_{\text{pl}}}\right) T^2 \,.
\ee
Then, we obtain the following Penrose inequality.
\be
\label{eqn:snecpi}
m_{\mathrm{ADM}} \geq \sqrt{\frac{\Amin(T)}{16\pi}}  \exp \left(-\frac{s}{4}\frac{2\sqrt{2\pi^3 }}{\sqrt{5}} \sqrt{B} \left(\frac{k}{\hbar T_{\text{pl}}}\right) T^2 \right) \,.
\ee

\begin{figure}
\centering     
\subfigure[Exponential factor varying with temperature ($B=1/3200 \pi$).]{\label{fig:a}\includegraphics[width=7cm]{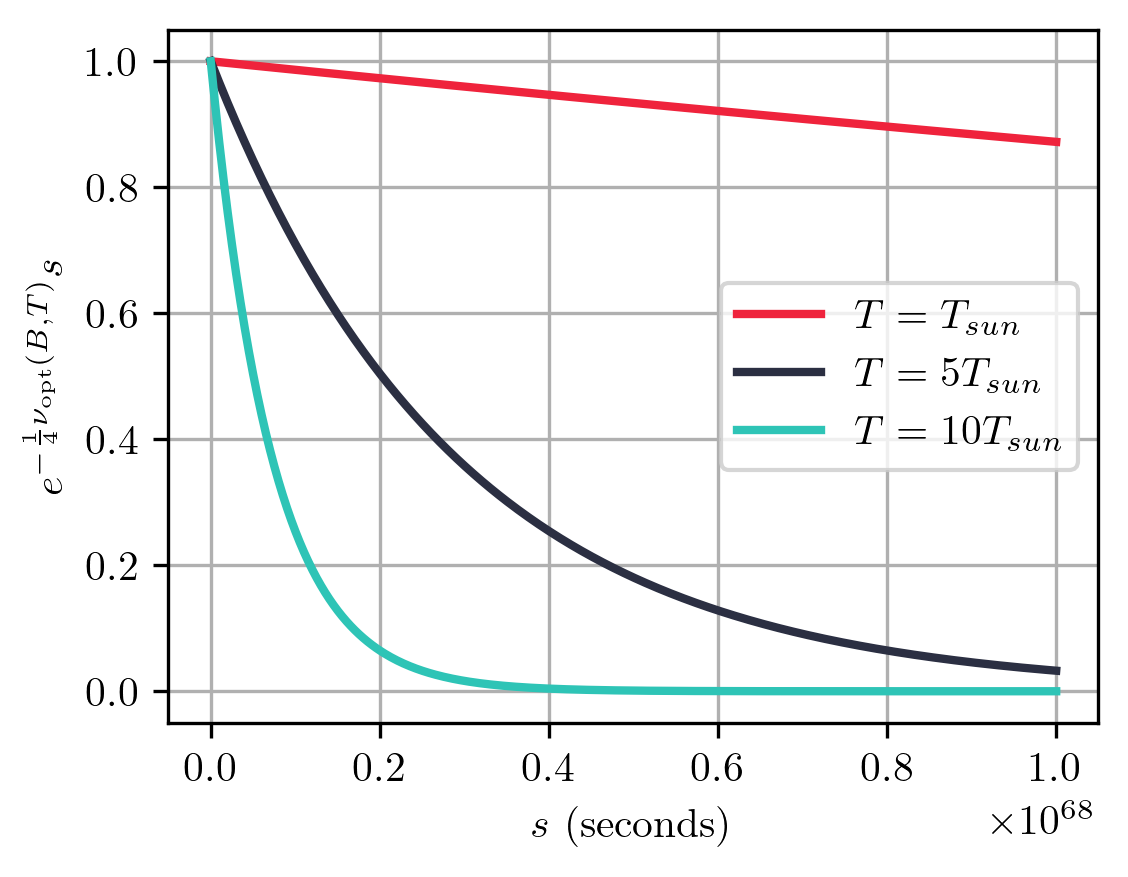}}
\subfigure[Exponential factor varying with the constant $B$ ($T=T_{\text{sun}}$).]{\label{fig:b}\includegraphics[width=7cm]{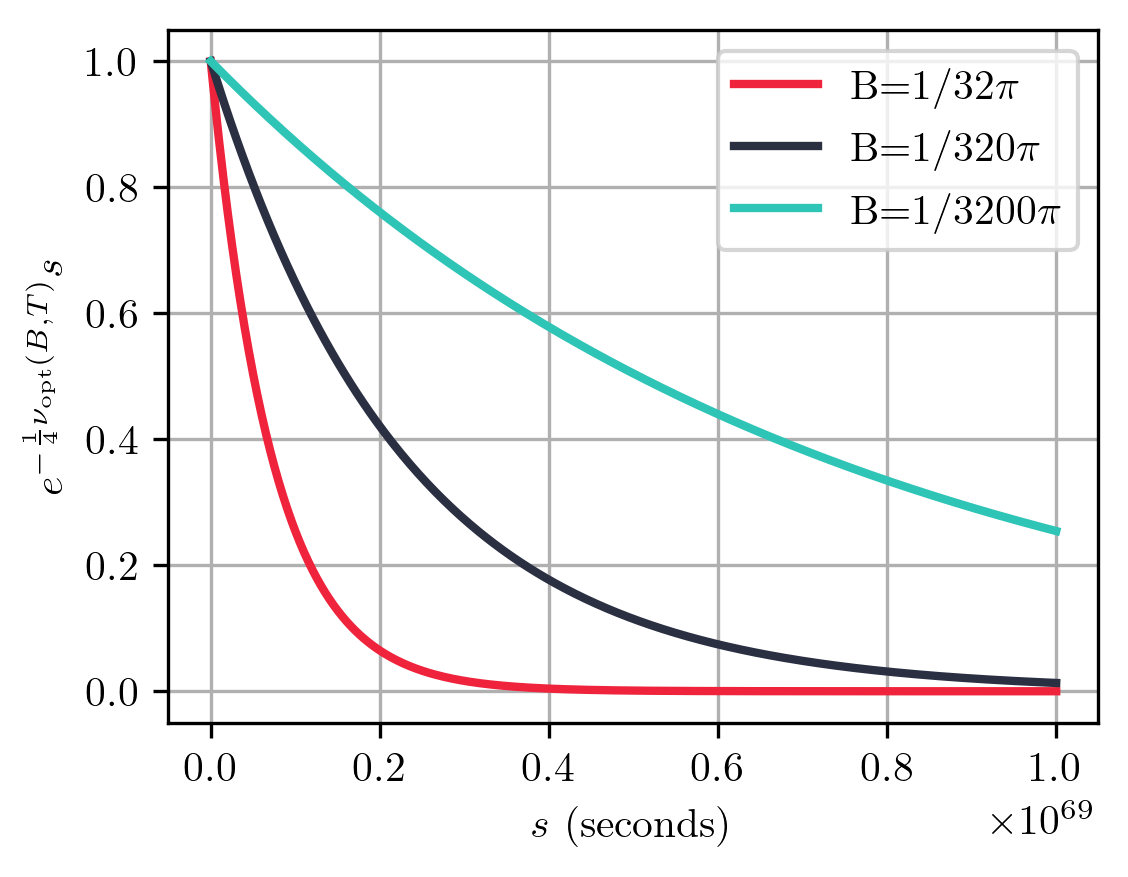}}
\caption{The figures show how much the inequality of Eq.~\eqref{eqn:snecpi} differs from the original Penrose one in terms of the affine distance of the two initial data sets ($s$). We see that with increasing temperature and $B$ (both meaning more negative energy allowed) the factor goes to zero faster. As reference temperature ($T_{\text{sun}}$) we are using the Hawking temperature of a solar mass black hole. Most astrophysical black holes would have a lower temperature.}
\label{fig:snec}
\end{figure}

To examine how this exponential factor differs from $1$ (the original Penrose inequality) we plot it versus $s$ for different temperatures and values of $B$ in Fig.~\ref{fig:snec}. From Eqs.~\eqref{eqn:snec} and \eqref{eqn:rhozerofin}, we see that larger values of $B$ and $T$ lead to more allowed negative energy. We see that this is translated to faster decay of the exponential to zero. We should note that we are using the assumption that the initial data set satisfies the S-DEC thus there is no possibility of negative ADM mass. 

\subsubsection{Use as a test of the weak cosmic censorship conjecture}

We conclude this section by briefly commenting how Theorem~\ref{the:PIevapo} could be used as a test for the weak cosmic censorship conjecture. First we recall how the global Penrose inequality \eqref{eqn:Penroseoriginal} could be used: one can imagine that the cross section of the event horizon and the ADM mass of a black hole spacetime are known but nothing else. In particular it is not known if the NCC holds. Then the relationship between the cross section of the event horizon and the ADM mass can be examined, and if it follows the Penrose inequality then no violation of the weak cosmic censorship conjecture is detected even though energy conditions were not assumed.

To do something similar with Theorem~\ref{the:PIevapo} we need an earlier expression of Eq.~\eqref{eqn:PIevap}, in particular
\be
m_{\mathrm{ADM}} \geq \sqrt{{\frac{|\mathcal{H} \cap M_0|}{ 16\pi}}} \exp \left(-\frac{1}{4} \nu_{\text{opt}} \cdot s \right)\,,
\ee
which follows from the proof of the theorem, assuming $4$-dimensions and the optimization of $\nu$ explained previously.

Now let's assume that one has the cross section of the event horizon of an evolving evaporating initial data set and the ADM mass of the spacetime but no energy conditions. To do a similar comparison we need to employ a QEI to estimate the value of $\nu$ over $s$. Then we would have a curve over $s$ similar to Fig.~\ref{fig:snec}. Testing if the evolution of the relationship between the ADM mass of the spacetime and the cross section of the event horizon is above or below that curve would give whether or not the modified Penrose inequality holds. If it does, no violation of the weak cosmic censorship conjecture is detected as the sufficiently trapped surface would be behind the event horizon. One disadvantage compared to the original Penrose inequality is that an energy condition is needed to determine the evolution of event horizon. This approach is also inviting the same criticism as Eq.~\eqref{eqn:Penroseoriginal}: the event horizon is very difficult to measure before the full evolution. Thus, even though QEIs for different fields exist, it would not be practical to test the expression against evaporating black hole spacetimes.

Due to that inability to know the future evolution of a spacetime, the original Penrose inequality is more useful in the form of Eq.~\eqref{eqn:Penroselocal} where the area of a trapped surface is instead used. The Eq.~\eqref{eqn:PIevap} could also be used as a test for the existence of an event horizon (as a consequence of the causal structure of spacetime). The idea is the following: one could examine a black hole model that includes a trapped surface leading to a singularity. Now an observation of violation of the inequality of Eq.~\eqref{eqn:PIevap} with an appropriate QEI would lead to the conclusion of possible non-existence of an event horizon. This follows as any sufficiently trapped surface with the correct causal structure would be behind the event horizon.

\section{Discussion}
\label{sec:discussion}

In this work we showed that the Penrose inequality, or a modified version of it, can be proven under averaged energy conditions. This was shown both in the case of the initial data set formulation and in the case of the global Penrose inequality where a connection with the weak cosmic censorship is made. The latter case applies for evaporating black holes where the area of two initial data sets at different times is compared. 

In both cases, the more negative energy allowed, the further the departure from the original Penrose inequality. However in the case of evaporating black holes, a version of the weak cosmic censorship is restored: using the idea of the sufficiently trapped surface, a surface that leads to null geodesic incompleteness, this surface is always behind an event horizon. Furthermore, it can lead to a modified Penrose inequality that can be used as a test of this notion.

For the modified Penrose inequality we use two energy conditions, a QEI--inspired condition that could be obeyed by quantum fields and the S-DEC. The theorem could be easily generalized for the S-ADEI, however neither the S-DEC or the S-ADEI is immediately obeyed by quantum fields. An issue that arises is that quantum fields do not generally obey spacelike integrated bounds \cite{FordHelferRoman2002}. A way forward is to consider only states that solve the semiclassical Einstein equation and generate asymptotically flat spacelike hypersurfaces. Bounds for those kinds of states have not been yet derived.

It would of course be an important extension of this work to prove a Penrose inequality with one QEI-inspired condition instead of two. However, these conditions are independent from one another. A possible approach could be a worldvolume bound, an integral of the dominant energy condition over a spacetime region. The Hawking singularity theorem has been proven \cite{Graf:2022mko} with such a condition but it is unclear if a version of the Penrose inequality could also be derived. 

The positive mass (or energy) theorem, stating the positivity of a version of the global mass, is closely related to the Penrose inequality and the weak cosmic censorship conjecture. There are various proofs, with most important ones by Schoen and Yau \cite{Schon:1979rg, Schon:1981vd} and by Witten \cite{Witten:1981mf}. Both versions and most subsequent proofs use the DEC. In addition, a positive mass theorem allowing pointwise negative scalar curvature was also established by \cite{LLR2022, LUY2024}, where negativity is compensated for by large pointwise positive scalar curvature on an annulus. Here we derive the positive mass theorem for spherically symmetric spacetimes that satisfy a version of the S-ADEI in Corollary~\ref{cor:lowerboundK}. So perhaps in more general spacetimes, it would be possible to prove the positive mass theorem without a pointwise condition but instead an average form of the DEC. This remains an objective for future work.

\vspace{-0.11in}

\begin{acknowledgments}
We thank Melanie Graf and Andrew Svesko for useful discussions. EH acknowledges the support of the German Research Foundation through the excellence
cluster EXC 2121 ``Quantum Universe'' – 390833306. 
\end{acknowledgments}

\newpage

\bibliographystyle{utphys}
\bibliography{references}

\end{document}